\documentclass[11pt]{article}
\usepackage{amsmath, amssymb, amsthm}
\usepackage{times}
\usepackage{tikz}
\usepackage{graphicx}
\usepackage{url}
\usepackage{caption}
\usepackage{subcaption}
\usepackage{nicefrac}
\usepackage[numbers]{natbib}
\usepackage[font={footnotesize,it}]{caption}
\usepackage{fullpage}
\usepackage{dsfont}

\newcommand{\bw}{\mathbf{w}}
\newcommand{\bx}{\mathbf{x}}
\newcommand{\bbR}{\mathbb{R}}

\newcommand{\U}{\mathcal{U}}
\newcommand{\given}{ \ : \ }
\newtheorem{definition}{Definition}
\newtheorem{theorem}{Theorem}
\newtheorem{lemma}{Lemma}
\newtheorem{corollary}{Corollary}

\newtheorem*{theorem*}{Theorem}
\newtheorem*{lemma*}{Lemma}

\newcommand{\CC}{\mathcal{C}}
\newcommand{\CD}{\mathcal{D}}

\newcommand{\CH}{\mathcal{H}}

\newcommand{\CX}{\mathcal{X}}
\newcommand{\CS}{\mathcal{S}}
\newcommand{\R}{\mathbb{R}}
\newcommand{\avg}{\texttt{avg}}
\providecommand{\norm}[1]{\left\lVert#1\right\rVert}
\DeclareMathOperator{\poly}{poly}
\DeclareMathOperator{\sign}{sign}
\DeclareMathOperator{\E}{E}
\newcommand{\EU}[1][]{\underset{#1}{\E}}

\newcommand{\scs}{\textsc{Statistical Cost Sharing}}

\usepackage{thmtools}
\usepackage{thm-restate}

\title{Statistical Cost Sharing}
\author{Eric Balkanski\footnote{School of Engineering and Applied Sciences, Harvard University, ericbalkanski@g.harvard.edu.} 
\qquad Umar Syed \footnote{Google NYC, usyed@google.com.} 
\qquad Sergei Vassilvitskii\footnote{Google NYC, sergeiv@google.com.}}

\date{}

\begin{document}

\maketitle

\begin{abstract}
We study the cost sharing problem for cooperative games in situations where the cost function $C$ is not available via oracle queries, but must instead be derived from data, represented as tuples $(S, C(S))$, for different subsets $S$ of players. We formalize this approach, which we call \scs,
and consider the computation of the core and the Shapley value, when the tuples are drawn from some distribution $\CD$. 

Previous work by \citet{BPZ-15} in this setting showed how to compute cost shares that  satisfy the core property with high probability for limited classes of functions. We expand on their work and give an algorithm that computes such cost shares for any function with a non-empty core.  We complement these results by proving an inapproximability lower bound for a weaker relaxation.

We then turn our attention to the Shapley value. We first show that when cost functions come from the family of submodular functions with bounded curvature, $\kappa$, the Shapley value can be approximated from samples up to a $\sqrt{1 - \kappa}$ factor, and that the bound is tight. We then define statistical analogues of the Shapley axioms, and derive a notion of statistical Shapley value. We show that these can always be approximated arbitrarily well for general functions over any distribution $\CD$. 
\end{abstract}

\newpage




\section{Introduction}
\label{sec:intro}
The cost sharing problem asks for an equitable way to split the cost of a service among all of the participants. Formally, there is a cost function defined over all subsets of a ground set of elements (or players) and the objective is to fairly divide the cost of the full set among the participants. Cost sharing is central to cooperative game theory, and there is a rich literature developing the key concepts and principles to reason about this topic. Two popular cost sharing concepts are the {\em core}~\cite{gillies1959solutions}, where no group of players has an incentive to deviate, and the {\em Shapley value}~\cite{shapley2016value}, which is the unique vector of cost shares satisfying four natural axioms.

While both the core and the Shapley value are easy to define, computing them poses additional challenges. One obstacle is that the computation of the cost shares requires knowledge of costs in myriad different scenarios. For example, computing the exact  Shapley value requires one to look at the marginal contribution of a player over {\em all possible subsets}. Recent work ~\cite{liben2012computing} shows that one can find approximate Shapley values for a restricted subset of cost functions by looking at the costs for polynomially many specifically chosen examples. In practice, however, another roadblock emerges: one cannot simply query for the cost of a hypothetical scenario. Rather, the costs for scenarios that have not occurred are simply unknown. We share the opinion of \citet{BPZ-15} that the main difficulty with using  cost sharing methods in concrete applications is the information needed to compute them.

Concretely, consider the following cost sharing applications. 
\paragraph{Attributing Battery Consumption on Mobile Devices.}
A modern mobile phone or tablet is typically  running a number of distinct apps concurrently. In addition to foreground processes, a lot of activity may be happening in the background: email clients may be fetching new mail, GPS may be active for geo-fencing applications,  messaging apps are polling for new notifications, and so on. All of these activities consume power; the question is how much of the total battery consumption should be attributed to each app? This problem is non-trivial because the operating system induces cooperation between apps to save battery power.  For example there is no need to activate the GPS sensor twice if two different apps request the current location almost simultaneously. 

\paragraph{Moneyball and Player Ratings} 
The impact of an individual player on the overall performance of the team typically depends on the other players currently playing. One can infer the total  benefit from the players on the field (or on the court) from metrics like number of points scored, time of possession, etc., the question here is how to allocate this impact to the individuals. Recently many such metrics have been proposed (for example plus/minus ratio in hockey, wins above replacement in baseball.), our goal here is to find scores compatible with cooperative game theory concepts.  

\paragraph{Understanding Black Box Learning} 
Deep neural networks are prototypical examples of black box learning, and it is almost impossible to tease out the contribution of a particular feature to the final output. Particularly, in situations where the features are binary, cooperative game theory gives a formal way to analyze and derive these contributions. While one can evaluate the objective function on any subset of features, deep networks are notorious for performing poorly on certain out of sample examples~\cite{AdversarialNN, AdversarialNN2}, which may lead to misleading conclusions when using traditional cost sharing methods.

We model these cost sharing questions as follows.  Let $N$ be the set of possible players (apps or features), and for a subset $S \subseteq N$, let $C(S)$ denote the cost of $S$. This cost represents the total power consumed over a standard period of time, or the number of points scored by the team, and so on. We are given ordered pairs $(S_1, C(S_1)), (S_2, C(S_2)), \ldots,$ $(S_m, C(S_m))$, where each $S_i \subseteq N$ is drawn independently from some distribution $\CD$.  The problem of \scs\ asks to look for reasonable cost sharing strategies in this setting. 

\subsection{Our results} 
We build on the approach from \citet{BPZ-15}, which studied \scs\ in the context of the core, and assume that only partial data about the cost function is observed. The authors showed that cost shares that are likely to respect the core property  can be obtained for certain restricted classes of functions.  Our main result is an algorithm that generalizes these results for {\em all} games where the core is non-empty and we derive sample complexity bounds showing exactly the number of samples required to compute cost shares (Theorems~\ref{t:vccore} and \ref{thm:approx_core_learn}).\footnote{Concurrently and independently of our work, \citet{BPZ-16} proved a polynomial sample complexity bound for this problem of computing cost shares that are likely to respect the core property, for all functions with a non-empty core.}  While the main approach of \citet{BPZ-15} relied on first learning the cost function and then computing cost shares, we show how to proceed directly, computing cost shares  without explicitly learning a good estimate of the cost function. We also show that approximately satisfying the core with probability one is impossible in general (Theorem~\ref{t:core}).

We then focus on the Shapley value, which has never been studied in the \scs \ context. We introduce a new cost sharing method called \emph{data-dependent Shapley value}  which is the unique solution (Theorem~\ref{t:ddshapley}) satisfying four natural axioms resembling the Shapley axioms (Definition~\ref{d:ddshapley}), and which can be approximated arbitrarily well from samples for any bounded function and any distribution (Theorem~\ref{t:apxddshapley}). Regarding the traditional Shapley value,  we obtain a tight $\sqrt{1 - \kappa}$ multiplicative approximation for submodular functions with bounded curvature $\kappa$ over the uniform distribution (Theorems~\ref{t:curv} and \ref{t:lbcurv}), but show that they cannot be approximated by a bounded factor in general, even for the restricted class of coverage functions, which are learnable,  over the uniform distribution (Theorem~\ref{thm:lower}).

\subsection{Related work}
There are two avenues of work which we build upon. The first is the notion of cost sharing in cooperative games, first introduced by  \citet{von2007theory}. We consider the Shapley value and the core, two popular solution concepts for cost-sharing in cooperative games. The Shapley value \cite{shapley2016value} is studied in algorithmic mechanism design \cite{anshelevich2008price,balkanski2015mechanisms,feigenbaum2000sharing,moulin1999incremental}. For applications of the Shapley value, see the surveys by \citet{roth1988shapley} and \citet{winter2002shapley}. A naive computation of the Shapley value of a cooperative game would take exponential time; recently, methods for efficiently approximating the Shapley value have been suggested \cite{bachrach2010approximating,fatima2008linear,liben2012computing,mann1960values} for some restricted settings.  

The core, introduced by \citet{gillies1959solutions}, is another well-studied solution concept for cooperative games. \citet{bondareva1963some} and \citet{shapley1967balanced} characterized when the core is non-empty. The core has been studied in the context of multiple combinatorial games, such as facility location \cite{goemans2004cooperative} and maximum flow \cite{deng1999algorithmic}. In cases with no solutions in the core  or when it is computationally hard to find one, the balance property has been relaxed to hold approximately \cite{devanur2005strategyproof, immorlica2008limitations}. In applications where players submit bids, cross-monotone cost sharing, a concept stronger than the core that satisfies the group strategy proofness property, has attracted a lot of attention~\cite{immorlica2008limitations,jain2002equitable,moulin2001strategyproof,pal2003group}. We note that these applications are sufficiently different from the ones we are studying in this work. 

The second is the recent work in econometrics and computational economics that aims to estimate critical concepts directly from a limited data set, and reason about the sample complexity of the computational problems.  Specifically, in all of the above papers, the algorithm must be able to query or compute $C(S)$ for an arbitrary set $S \subseteq N$. In our work, we are instead given a collection of samples from some distribution; importantly the algorithm does not know $C(S)$ for sets $S$ that were not sampled. This approach was first introduced by \citet{BPZ-15}, who showed how to compute an approximate core for some families of games. Their main technique  is to first learn the cost function $C$ from samples and then to use the learned function to compute cost shares. The authors also showed that there exist games that are not PAC-learnable but that have an approximate core that can be computed.

Outside of cooperative game theory, this data-driven approach has attracted renewed focus. In auction design, a line of work \cite{chawla2014mechanism,cole2014sample, dughmi2014sampling, morgenstern2015pseudo} has studied revenue maximization from samples instead of being given a Bayesian prior. In the inductive clustering setting, the algorithm is only given a small random subset of the data set it wishes to cluster \cite{balcan2009finding,balcan2009agnostic}. More closely related to our work is the problem of optimization from samples \cite{ curvature, BRS17} where,  the goal is to approximate $\max_{S \in M} C(S)$ for some constraint $M \subseteq 2^N$ and $C : 2^N \rightarrow \bbR$ from samples for some class of combinatorial functions.


\section{Preliminaries}
A \emph{cooperative game} is defined by an ordered pair $(N, C)$, where $N$ is the ground set of \emph{elements}, also called \emph{players}, and $C : 2^N \rightarrow \mathbb{R}_{\geq 0}$ is the \emph{cost function} mapping each \emph{coalition} $S \subseteq N$ to its cost, $C(S)$. The ground set of size $n = |N|$ is called the \emph{grand coalition} and we denote the elements by $N = \{1, \ldots, n\} = [n]$. We assume  that $C(\emptyset) = 0$, $C(S) \ge 0$ for all $S \subseteq N$, and that $\max_S C(S)$ is bounded by a polynomial in $n$, which are standard assumptions. We will slightly abuse notation and use $C(i)$ instead of $C(\{i\})$ for $i \in N$ when it is clear from the context. 

We recall three specific classes of functions. \emph{Submodular} functions exhibit the property of diminishing returns: $C_S(i) \geq C_T(i)$ for all $S \subseteq T \subseteq N$ and $i \in N$ where $C_S(i)$ is the marginal contribution of element $i$ to set $S$, i.e., $C_S(i) = C(S \cup \{i\}) - C(S)$. 
\emph{Coverage} functions are the canonical example of submodular functions. A function is coverage if it can be written as $C(S) = | \cup_{i \in S} T_i|$ where $T_i \subseteq U$ for some universe $U$. 
Finally, we also consider the simple class of additive functions that are such that $C(S) = \sum_{i\in S}C(i)$.

A \emph{cost allocation} is a vector $\psi \in \R^n$ where $\psi_i$ is the \emph{share} of element $i$. 
We call a cost allocation $\psi$ \emph{balanced} if $\sum_{i \in N} \psi_i = C(N)$. Given a cooperative game $(N, C)$ the goal in the cost sharing literature is to find ``desirable" balanced cost allocations. Most proposals take an axiomatic approach, defining a set of axioms that a cost allocation should satisfy. These lead to the concepts of Shapley value and the core, which we  define next. A useful tool to describe and compute these cost sharing concepts is permutations. We denote by $\sigma$ a uniformly random permutation of $N$ and by $S_{\sigma < i}$ the players before $i$ in permutation $\sigma$.

\subsection{The core}
The core is a balanced cost allocation where no player has an incentive to deviate from the grand coalition---for any subset of players the sum of their shares does not cover their collective cost. 
\begin{definition} A cost allocation $\psi$ is in the \emph{core} of function $C$ if the following properties are satisfied:
\begin{itemize}
\item \textbf{Balance:} $\sum_{i \in N} \psi_i = C(N)$, 
\item \textbf{Core property:} for all $S \subseteq N$, $\sum_{i \in S} \psi_i \leq C(S)$.
\end{itemize}
\end{definition}
The core is a natural cost sharing concept. For example, in the battery blame scenario it translates to the following assurance:  No matter what other apps are running concurrently, an app is never blamed for more battery consumption than if it were running alone. Given that app developers are typically business competitors, and that a mobile device's battery is a very scarce resource, such a guarantee can rather neatly avoid a great deal of finger-pointing.  Unfortunately, for a given cost function $C$ the core may not exist (we say the core is empty), or there may be multiple (or even infinitely many) cost allocations in the core. For submodular functions $C$, the core is guaranteed to exist and one allocation in the core can be computed in polynomial time. Specifically, for any permutation $\sigma$, the cost allocation $\psi$ such that  $\psi_i = C(S_{\sigma < i} \cup \{i\}) - C(S_{\sigma < i})$ is in the core. 


\subsection{The Shapley value}
The Shapley value provides an alternative cost sharing method. For a game $(N, C)$ we denote it by $\phi^C$, dropping the superscript when it is clear from the context. 
  While the Shapley value may not satisfy the core property, they satisfy the following four equally natural axioms: 

\begin{itemize}
\item \textbf{Balance:} $\sum_{i \in N} \phi_i  = C(N)$.
\item \textbf{Symmetry:} For all $i, j \in N$, if $C(S \cup \{i\}) = C(S \cup \{j\})$ for all $S \subseteq N \setminus \{i, j\}$ then $\phi_i = \phi_j$.
\item \textbf{Zero element:} For all $i \in N$, if $C(S \cup \{i\}) = C(S)$ for all $S \subseteq N$ then $\phi_i = 0$.
\item \textbf{Additivity:}  For two games $(N, C_1)$ and $(N, C_2)$ with the same players, but different cost functions $C_1$ and $C_2$, let $\phi^1$ and $\phi^2$ be the respective cost allocations. Consider a new game $(N, C_1 + C_2)$, and let $\phi'$ be the cost allocation for this game. Then for all elements, $i \in N$, $\phi'_i = \phi^1_i + \phi^2_i$.
 
\end{itemize}
Each of these axioms is natural: balance ensures that the cost of the grand coalition is distributed among all of the players. Symmetry states that two identical players should have equal shares. Zero element verifies that a player that adds zero cost to any coalition should have zero share. Finally, additivity just confirms that costs combine in a linear manner.  It is surprising that the set of cost allocations that satisfies all four axioms is unique. Moreover, the Shapley value $\phi$ can be written as the following summation: 

$$ \phi_i  = \EU[\sigma][C(S_{\sigma < i} \cup\{i\}) - C(S_{\sigma < i})] = \sum_{S \subseteq N \setminus\{i\}}\frac{|S|!(n - |S| - 1)!}{n!} (C(S \cup\{i\}) - C(S)).$$

This expression is the expected marginal contribution $C(S \cup\{i\}) - C(S)$ of $i$ over a set of players $S$ who arrived before $i$ in a random permutation of $N$.  As the summation is over exponentially many terms, the Shapley value generally cannot be computed exactly in polynomial time. However, several sampling approaches have been suggested to approximate the Shapley value for specific classes of functions~\cite{bachrach2010approximating,fatima2008linear,liben2012computing,mann1960values}. 

\subsection{Statistical cost sharing}
With the sole exception of ~\cite{BPZ-15}, previous work in cost-sharing critically assumes that the algorithm is given oracle access to $C$, i.e., it can query, or determine, the cost $C(S)$ for any $S \subseteq N$. In this paper, we aim to (approximately) compute the Shapley value and other cost allocations from \emph{samples}, without oracle access to $C$, and with a number of samples that is polynomial in $n$.

\begin{definition}\label{def:scs}
Consider a cooperative game with players $N$ and cost function $C$. In the \scs\ problem we are given pairs $(S_1, C(S_1)), (S_2, C(S_2)), \ldots, (S_m, C(S_m))$ where each $S_i$ is drawn i.i.d. from a distribution $\CD$ over $2^N$. The goal is to find a cost allocation $\psi \in \R^n$.  
\end{definition}

In what follows we will often refer to an individual $(S, C(S))$ pair as a {\em sample}. It is tempting to reduce \scs\ to classical cost sharing by simply collecting enough samples to use known algorithms. For example, \citet{liben2012computing} showed how to approximate the Shapley value with polynomially many queries $C(S)$. However, if the distribution $\CD$ is not aligned with these specific queries, which is the case for the uniform distribution, emulating these algorithms in our setting requires exponentially many samples.  \citet{BPZ-15} showed how to instead first learn an approximation to $C$ from the given samples and then compute cost shares for the learned function, but their results hold only for a limited number of games and cost functions $C$.  We show that a more powerful approach is to compute cost shares directly from the data, without explicitly learning the cost function first. 


\subsection{Warm up: linear functions and product distributions}

As a simple example, we consider the special case of additive functions with $C(i) \geq 1 / \poly(n)$ and on bounded product distributions $\CD$.\footnote{A bounded product distribution has marginal probabilities bounded below and above by $1/\poly(n)$ and $1 - 1/\poly(n)$.} In this setting, the core property and the Shapley value can be  approximated arbitrarily well. It is easy to verify that the cost allocation $\psi$ such that $\psi_i = C(i)$ is in the core and that it is the Shapley value.

To compute these cost shares, we 
estimate the expected marginal contribution of an element $i$ to a random set, i.e.,  $v_i^{\CD} := \E_{S \sim \CD|i \not \in S}[C(S \cup \{i\}) - C(S)]$. 
Note that in the case of additive functions, $v_i^{\CD} = C(i)$. In addition,
$$  v_i^{\CD} = \EU[S \sim \CD|i \not \in S][C(S \cup \{i\}) - C(S)] = \EU[S \sim \CD|i  \in S][C(S)] - \EU[S \sim \CD|i \not \in S][C(S)],$$
when $\CD$ is a product distribution. Thus, this value can be estimated arbitrarily well by looking at the difference in cost between the average value of the samples containing $i$ and the average of those not containing $i$. The analysis is a simple concentration bound and is deferred to the appendix.

 \begin{restatable}{rLem}{lconcentration}
 \label{l:addconcentration}
 Let $C$ be an additive function with $C(i) \geq 1 / \poly(n)$ for all $i$. Then, given $\poly(n, \nicefrac{1}{\delta}, \nicefrac{1}{\epsilon})$ samples, we can compute an estimate  $\tilde{v}_i$ such that with probability $(1 - \delta)$: 
   $$ (1 - \epsilon) C(i) < \tilde{v}_i < (1+\epsilon) C(i).$$
\end{restatable}

Thus an algorithm which computes $\psi_i =\tilde{v}_i$  approximates the Shapley value and the core property arbitrarily well (the formal definitions of approximating the core and the Shapley value are deferred to the respective sections devoted to those concepts).
%

\section{Approximating the Core from Samples}
\label{s:corefromsamples}

In this section,  we consider the problem of finding cost allocations from samples that satisfy relaxations of the core. A natural approach to this problem is to first learn the underlying model, $C$, from the data and to then compute a cost allocation for the learned function. As shown in \cite{BPZ-15}, this approach works if $C$ is PAC-learnable, but there exist functions $C$ that are not PAC-learnable and for which a cost allocation that approximately satisfies the core can still be computed. The main result of this section shows that a cost allocation that  approximates the core property can be computed for \emph{any} function with a non-empty core. Moreover, we show that the number of samples from $\CD$ needed to accurately learn the core is low. 

The approach is to directly compute a cost allocation, which empirically satisfies the core property, i.e., it satisfies the core property on all of the samples. We then argue that the same cost shares will likely satisfy the core property on newly drawn samples as well. This  difference between the empirical performance of a function and its expected performance is known as the generalization error of a function and its analysis is central to theoretical machine learning. Intuitively, the generalization error is small when the number of samples, $m$, used to train the function is large, and the function itself is relatively simple. Two of the most common tools that formally capture these notions are the VC-dimension and the $m$-sample Rademacher complexity of a function class. We will use both of these to highlight different trade-offs in computing statistical cost shares. 

We begin by defining three notions of approximate core: the probably stable \cite{BPZ-15}, approximately stable, and probably approximately stable cores. 

\begin{definition} Given $\delta, \epsilon > 0$, a cost allocation $\mathbf{\psi}$ such that $\sum_{i \in N} \psi_i = C(N)$ is in
\begin{itemize}
\item the \textbf{probably stable} core \cite{BPZ-15} if, for all $\CD$, $$\Pr_{S \sim \CD} \left[ \sum_{i \in S} \psi_i \leq C(S) \right] \geq 1 - \delta,$$
\item the \textbf{approximately stable} core over $\CD$ if for all $S \subseteq N$, $$(1- \epsilon) \cdot \sum_{i \in S} \psi_i \leq C(S), $$ 
\item the \textbf{probably  approximately stable} core if, for all $\CD$, $$\Pr_{S \sim \CD} \left[ (1- \epsilon)\cdot \sum_{i \in S} \psi_i \leq C(S) \right] \geq 1 - \delta.$$
\end{itemize} 
\end{definition}

The algorithms we consider compute cost shares in polynomial time. The hardness results are information theoretic and are not due to running time limitations.

\begin{definition}
Cost shares $\mathbf{\psi}$ are  \emph{computable} for the class of functions $\CC$ over distribution $\CD$, if for all $C \in \mathcal{C}$ and any $\Delta, \delta, \epsilon > 0$, given $C(N)$ and $m = \poly(n, \nicefrac{1}{\Delta}, \nicefrac{1}{\delta}, \nicefrac{1}{\epsilon})$ samples $(S_j, C(S_j))$ with each $S_j$ drawn i.i.d. from distribution $\mathcal{D}$, there exists an algorithm that computes $\mathbf{\psi}$ with probability at least $1-\Delta$ over both the samples and the choices of the algorithm. If the algorithm has $\poly(m)$ running time, then the cost shares $\mathbf{\psi}$ are  \emph{efficiently computable}.
\end{definition}

Finally, we will refer to the number of samples $m$ required to compute approximate cores as the {\em sample complexity} of the algorithm. 

\paragraph{Our Results.} We give algorithms that efficiently compute cost shares in the probably stable core for  functions with a non-empty (traditional) core with  a simple approach using the VC-dimension (Section~\ref{s:probstable}), the algorithm has sample complexity linear in $n$. With a more complex analysis and using the Rademacher complexity, we obtain efficiently computable cost shares in the probably approximately stable core with an improved sample complexity dependence of $\log n$ but with an additional dependence on the spread of the function $C$ (Section~\ref{s:probapxstable}). Finally, we show that  cost shares in the approximately stable core are not computable even for the uniform distribution and the well-behaved class of monotone submodular functions (Section~\ref{s:apxstable}).

\subsection{Cost shares in the probably stable core are efficiently computable}
\label{s:probstable}
\citet{BPZ-15} showed that several families of functions $\CD$ have a core that is probably stable. These families include network flow, threshold task, and induced subgraph games which are all well-known classes of cooperative games; and the class of monotone simple games which are games that take values in $\{0,1\}$. We generalize their result so that it is not constrained to specific classes of functions, and show how to compute a probably stable core for any game with a non-empty core. 

Technically, we use the VC-dimension of the class of halfspaces to show that the performance on the samples generalizes well to the performance on the distribution $\CD$. We review the definition of the VC-dimension in Appendix~\ref{s:appcore} and only state VC-dimension results needed for our purposes. We first state the generalization error obtained for a class of functions with VC-dimension $d$.
 
\begin{theorem}[\cite{ML}, Theorem 6.8]
\label{t:vcmain}
Let $\mathcal{H}$ be a hypothesis class of functions from a domain $\mathcal{X}$ to $\{-1, 1\}$ and $f : \mathcal{X} \mapsto \{-1, 1\}$ be some ``correct" function. Assume that $\CH$ has VC-dimension $d$. Then, there is an absolute constant $c$ such that with  $m \geq c (d  + \log(1/ \Delta))/\delta^2$ i.i.d. samples $\bx^1, \ldots, \bx^m \sim \CD$, 
$$\left|\Pr_{\bx \sim \CD}\left[h(\bx) \neq f(\bx)\right] - \frac{1}{m}\sum_{i=1}^m  \mathds{1}_{h(\bx^i) \neq f(\bx^i)}\right| \leq \delta$$
for all $h \in \mathcal{H}$, with probability $1- \Delta$ over the samples.
\end{theorem} 
 
We use a special case of the class of halfspaces, for which we know the VC-dimension. 
 
\begin{theorem}[\cite{ML}, Theorem 9.2]
\label{t:vchalfspace}
The class of functions $\{\bx \mapsto \sign(\bw^\intercal \bx) : \bw \in \mathbb{R}^n\}$ has VC-dimension $n$. 
\end{theorem}

We first define a class of functions that contains the core, and prove that it has low VC-dimension.  Given a sample $S$, define $\bx^S$ such that $x^S_i = \mathds{1}_{i \in S}$ for $i \in [n]$ and $x^S_{n+1} = C(S)$. Note that if the core property is satisfied for sample $S$, then $\sign\left(\sum_{i = 1}^n \psi_i x^S_i - x^S_{n+1}\right) = -1$ . We  now bound the VC-dimension of this hypothesis class of functions induced by cost allocations $\psi$.
\begin{corollary}
The class of functions $\mathcal{H}^{core} = \{\bx \mapsto sign(\sum_{i = 1}^n \psi_i x_i - x_{n+1}) \given \psi \in \mathbb{R}^n, \sum_i \psi_i = C(N)\}$ has VC-dimension at most $n+1$. 
\end{corollary}
\begin{proof}
We combine the observation that $\{\bx \mapsto \sign(\sum_{i = 1}^n w_i x_i - x_{n+1}) \given \bw \in \mathbb{R}^n, \sum_i w_i = C(N)\} \subseteq \{\bx \mapsto \sign(\bw^\intercal \bx) \given \bw \in \mathbb{R}^{n+1}\}$ with the well-known fact that the VC-dimension of $\mathcal{H'}$ is at most the VC-dimension of $\CH$ for  $\mathcal{H'} \subseteq \mathcal{H}$. 
\end{proof}

It remains to show how to optimize over functions in this class. 

\begin{theorem}
\label{t:vccore}
 The class of functions with a non-empty core has cost shares in the probably stable core that are efficiently computable.
The sample complexity  is $$O\left(\frac{n  + \log(1/ \Delta)}{\delta^2}\right).$$
\end{theorem}
\begin{proof}
Let $\psi$ be a cost allocation which satisfies both the core property on all the samples and the balance property, i.e., $\sum_{i \in S} \psi_i \leq C(S)$ for all samples $S$ and $\sum_{i \in N} \psi_i = C(N)$. Note that such a cost allocation exists since we assume that $C$ has a non-empty core. Given the set of samples, it can be computed with a simple linear program. We argue that  $\psi$ is probably stable.

Define  $h(x) = \sign\left(\sum_{i = 1}^n \psi_i x^S_i - x^S_{n+1}\right)$  and $f(\bx) = - 1$ for all $\bx$. Since the core property is satisfied on all the samples, $\frac{1}{m}\sum_{i=1}^m  \mathds{1}_{h(x) \neq f(x)} = 0$. Thus, by Theorem~\ref{t:vcmain},
\begin{align*}
\Pr_{S \sim \CD}\left[ \sum_{i \in S} \psi_i \leq C(S) \right] & = 1 - \Pr_{\bx^S : S \sim \CD}\left[\sign\left(\sum_{i = 1}^n \psi_i x^S_i - x^S_{n+1}\right) \neq -1\right] \\
& = 1 - \Pr_{\bx^S : S \sim \CD}\left[h\left(\bx^S\right) \neq f\left(\bx^S\right)\right] \\
 & = 1 - \left|\Pr_{\bx^S : S \sim \CD}\left[h\left(\bx^S\right) \neq f\left(\bx^S\right)\right] - \frac{1}{m}\sum_{i=1}^m  \mathds{1}_{h\left(\bx^S\right) \neq f\left(\bx^S\right)}\right| \\
& \geq 1 - \delta
\end{align*}
with $O((n  + \log(1/ \Delta))/\delta^2)$ samples.
\end{proof}

\subsection{Logarithmic sample complexity for probably approximately stable cores}
\label{s:probapxstable}

We show that the sample complexity dependence on $n$  can be improved from linear to logarithmic.  However, this improvement comes at a cost. We now find a probably approximately stable core instead of probably stable core, and the sample complexity depends on the {\em spread} of the function $C$, defined as $\frac{\max_S C(S)}{\min_{S \not= \emptyset} C(S)}$. This approach assumes that $\min_{S \not= \emptyset} C(S) > 0$. 
We start with an overview.




\begin{enumerate}
\item As previously, we find a cost allocation which satisfies the core property on all samples. However, we restrict this search to cost allocations with bounded $\ell_1$-norm. Such a cost allocation can be found efficiently since the space of such cost allocations is convex.
\item The analysis begins by bounding the $\ell_1$-norm of any vector in the core (Lemma~\ref{l:boundedcore}). Combined with the assumption that the core is non-empty, this implies that a cost allocation $\psi$ satisfying the previous conditions exists.
\item Let $[x]_+$ denote the function $x \mapsto \max(x, 0)$. Consider the following ``loss" function:
$$ \left[\frac{\sum_{i \in S} \psi_i}{C(S)} - 1\right]_+$$
This loss function is convenient  since it is equal to $0$ if and only if  the core property is satisfied for $S$ and it is $1$-Lipschitz, which is used in the next step.

  \item  Next, we bound the difference between the empirical loss and the expected loss for all $\psi$ with a known result using the Rademacher complexity of linear predictors with low $\ell_1$ norm over $\rho$-Lipschitz loss functions (Theorem~\ref{thm:rad_main}).
  
  \item Finally, given $\psi$ which approximately satisfies the core property in expectation, we show that $\psi$ is in the probably approximately stable core  by Markov's inequality (Lemma~\ref{lem:markov}). 
  \end{enumerate}

We review the definition of the Rademacher complexity in Appendix~\ref{s:appcore}. For our purposes, the following result which follows from the Rademacher complexity of linear classes is sufficient.

\begin{theorem}[\cite{ML}, Theorem 26.15]
\label{thm:rad_main}
Suppose that $\CD$ is a distribution over $\CX \times \R$ such that with probability $1$ we have that $\|\bx\|_{\infty} \leq R.$ Let $\CH = \{\bw \in \R^d : \|\bw\|_1 \leq B\}$ and let $\ell : \CH \times (\CX \times \R) \rightarrow \R$ be a loss function of the form 
$$\ell(\bw, (\bx, y)) = \phi(\bw^{\intercal} \bx, y)$$
such that for all $y \in \R$, $a \mapsto \phi(a, y)$ is an $\rho$-Lipschitz function and such that $\max_{a \in [-BR, BR]}|\phi(a,y)| \leq c$. Then, for all $\bw \in \CH$ and any $\Delta \in (0,1)$, with probability of at least $1 - \Delta$ over  $m$  i.i.d. samples $(\bx_1, y_1), \ldots, (\bx_m, y_m)$ from $\CD$,
$$\EU[(\bx, y) \sim \CD][\ell(\bw, (\bx, y))] \leq   \frac1m \sum_{i=1}^m \ell(\bw, (\bx_i, y_i))+ 2\rho B R \sqrt{\frac{2 \log(2d)}{m} }  + c \sqrt{\frac{2 \log(2/\Delta)}{m}}.$$
\end{theorem}



We first bound the $\ell_1$ norm of vectors in the core to bound the space of linear functions that we search over.
\begin{lemma}
\label{l:boundedcore}
Assume that $\psi$ is a vector in the core, then $\|\psi\|_1 \leq 2 \max_S |C(S)|$. 
\end{lemma}
\begin{proof}
Fix some vector $\psi$ in the core. Let $A$ be the set of elements $i$ such that $\psi_i \geq 0$ and $B$ be the remaining elements. First note that $$\sum_{i \in A} \psi_i \leq C(A) \leq \max_S |C(S)|$$ where the first inequality is by the core property. Next, note that $$0  \leq C(N) = \sum_{i \in A} \psi_i + \sum_{i \in B} \psi_i \leq \max_S |C(S)| + \sum_{i \in B} \psi_i$$ where the  equality is by the balance property, so $\sum_{i \in B} \psi_i \geq -\max_S |C(S)|.$ Thus, $$\|\psi\|_1 = \sum_{i \in A} \psi_i - \sum_{i \in B} \psi_i \leq \max_S |C(S)| + \max_S |C(S)|.$$
\end{proof}

We can thus focus on bounded cost allocations $\psi \in \CH$ where
$$\CH := \left\{\psi \given \psi \in \mathbb{R}^n, \|\psi\|_1 \leq 2 \max_S |C(S)|\right\}.$$



The next lemma shows that if the core property approximately holds in expectation, then it is likely to approximately hold.

\begin{lemma} \label{lem:markov} For any $0 < \epsilon, \delta < 1$ and cost allocation $\psi$,  $$\EU[S \sim \CD]\left[\left[\frac{\sum_{i \in S} \psi_i}{C(S)} - 1\right]_+\right] \le \frac{\epsilon\delta}{1 - \epsilon} \Rightarrow  \Pr_{S \sim \CD}\left[(1 - \epsilon) \sum_{i \in S} \psi_i \le C(S)\right] \ge 1 - \delta.$$\end{lemma}

\begin{proof} For any $a > 0$ and nonnegative random variable $X$, by Markov's inequality we have $\Pr[X \le a] \ge 1 - E[X]/ a$. By letting $a = \epsilon /(1 - \epsilon)$, $X = \left[(\sum_{i \in S} \psi_i)/ C(S) - 1\right]_+$, and observing that
$$
\left[\frac{\sum_{i \in S} \psi_i}{C(S)} - 1\right]_+ \le \frac{\epsilon}{1 - \epsilon} \Rightarrow \frac{\sum_{i \in S} \psi_i}{C(S)} - 1 \le \frac{\epsilon}{1 - \epsilon} \Rightarrow (1 - \epsilon) \sum_{i \in S} \psi_i \le C(S),
$$
we obtain $\Pr_{S \sim \CD}\left[(1 - \epsilon) \sum_{i \in S} \psi_i \le C(S)\right] \ge 1 - \delta$.
\end{proof}

Combining Theorem~\ref{thm:rad_main}, Lemma~\ref{l:boundedcore}, and Lemma~\ref{lem:markov}, the dependence of the sample complexity  on $n$ is improved. 
\begin{theorem} \label{thm:approx_core_learn} The class of functions with a non-empty core has cost shares in the probably approximately stable core that are efficiently computable with sample complexity
$$
\left(\frac{1 - \epsilon}{\epsilon\delta}\right)^2 \left(128 \tau(C)^2 \log(2 n) + 8 \tau(C)^2 \log (2/\Delta) \right) = O\left(\left(\frac{\tau(C)}{\epsilon \delta}\right)^2 \left(\log n +\log(1/\Delta)\right)\right).
$$
where $\tau(C) =  \frac{ \max_S C(S)}{\min_{S \neq \emptyset} C(S)}$ is the spread of $C$.
\end{theorem}
\begin{proof}Fix $C \in \CC$. Suppose we are given $m$ samples from $\CD$. 
We pick  $\psi^{\star} \in \CH$ such that  core property holds on all the samples and such that the balance property holds ($\sum_{i \in N} \psi_i = C(N)$). This cost allocation $\psi^{\star}$ can be found efficiently since the collection of such $\psi$ is a convex set. By the assumption that $C$ has at least one vector in the core and by Lemma~\ref{l:boundedcore}, such a $\psi^\star$ exists. Given $S \sim \CD$, define $\bx^{S}$ such that $x^S_i = \mathds{1}_{i \in S} / C(S)$. Fix $y = 1$.  Define the loss function $\ell$ as follows,
$$\ell\left(\psi, \left(\bx^S, y\right)\right) := \left[\psi^{\intercal} \bx^S - y\right]_+ = \left[\frac{\sum_{i \in S} \psi_i}{C(S)} - 1\right]_+$$

We wish to use Theorem~\ref{thm:rad_main} with $R = 1 / \min_{S \neq \emptyset} |C(S)|$, $B = 2 \max_S |C(S)|$, $\phi(a,y) = \left[a - y\right]_+$, $\rho = 1$, and $c = \tau(C)$. We verify that all the conditions hold. First note that without loss of generality, samples where $S = \emptyset$ can be ignored, so $\norm{\bx^S}_\infty \le 1 / \min_{S \neq \emptyset} |C(S)|$ for all $S$. Next, $\|\psi\|_1 \leq 2 \max_S |C(S)|$ for $\psi \in \CH$ by definition of $\CH$. The loss function $\ell$ is of the form $\ell(\psi, (\bx, y)) = \phi(\psi^{\intercal} \bx, y) = \left[\psi^{\intercal} \bx - y\right]_+$ such that
$a \mapsto \phi(a,y) =  \left[ a - y\right]_+ $ is an $1$-Lipschitz function and such that $\max_{a \in [-BR, BR]} |\phi(a,y) | \leq 2 \max_S |C(S)|/ \min_{S \neq \emptyset} |C(S)| = 2 \tau(C)$. In addition, note that 
$$\frac1m \sum_{i=1}^m \ell\left(\psi^{\star}, \left(\bx^{S_i}, 1\right)\right) = 0$$
since the core property holds on all the samples. Thus, by Theorem~\ref{thm:rad_main},

$$\EU[S \sim \CD]\left[\frac{\sum_{i \in S} \psi^{\star}_i}{C(S)} - 1\right]_+ = \EU[\bx^S : S \sim \CD]\left[\ell\left(\psi^{\star}, \left(\bx^S, 1\right)\right)\right] \leq  4 \tau(C) \sqrt{\frac{2 \log(2n)}{m} }  + \tau(C) \sqrt{\frac{2 \log(2/\Delta)}{m}}.$$


Choose any $\epsilon, \delta > 0$. If the number of samples $m$ is chosen as in the statement of the theorem, then the righthand side of the above inequality will be less than $\frac{\epsilon\delta}{1 - \epsilon}$. Thus by Lemma \ref{lem:markov},
$$
 \Pr_{S \sim \CD}\left[(1 - \epsilon) \sum_{i \in S} \psi_i \le C(S)\right] \ge 1 - \delta,
$$
 which completes the proof.\end{proof}

\subsection{The approximately stable core is not computable}
\label{s:apxstable}
Since we obtained a probably approximately stable core, a natural question is  if it is possible to compute cost allocations that are approximately stable over natural distributions. The answer is negative in general: even for the restricted class of monotone submodular functions, which always have a solution in the core, the core cannot be approximated from samples, even over the uniform distribution. 
 \begin{restatable}{rThm}{tcore}
\label{t:core}
Cost shares $\psi$ in the $(1/2 + \epsilon)$-approximately stable core, i.e., such that for all $S$,
$$\left(\frac{1}{2} + \epsilon\right) \cdot \sum_{i \in S} \psi_i \leq C(S), $$ 
 cannot be computed for monotone submodular functions over the uniform distribution, for any constant $\epsilon > 0$.
\end{restatable}
\begin{proof}
The ground set of elements is partitioned in $\epsilon^{-1}$ sets $A_1, \ldots A_{\epsilon^{-1}}$  of size  $\epsilon n$  for some small constant $\epsilon > 0$. Let $\CC = \{C^{A_i} \given i \in [\epsilon^{-1}]\}$ where
$$C^A(S) = |(N \setminus A) \cap S| + \min \left(|A \cap S|, (1 + \epsilon) \frac{\epsilon n}{2}\right).$$
The expected number of elements of $A_i$ in a sample $S$ from the uniform distribution is $|A_i| / 2 = \epsilon n /2$, so by the Chernoff bound
$$\Pr \left[|A_i \cap S| \geq (1 + \epsilon)\frac{n\epsilon}{2}\right] \leq e^{-\frac{\epsilon^2 n }{6}},$$
Thus, by a union bound, $C^{A_i}(S) = |S|$ over all $i$ and all samples $S$ with probability $1 - O(e^{- n})$ and we henceforth assume this is the case. It is therefore impossible to learn any information about the partition  $A_1, \ldots A_{\epsilon^{-1}}$ from samples. Any cost allocation $\psi$ computed by an algorithm given samples from $C^{A_i}$ is  thus \emph{independent} of $i$. 

Next, consider such a cost allocation $\mathbf{\psi}$ independent of $i$ satisfying the balance property. There exists $A_i$ such that $\sum_{j \in A_i} \psi_j > (1 - \epsilon) n  \epsilon$ since $\sum_{j \in N} \psi_j = C(N) > (1 - \epsilon)n$ by the balance property. In addition, $C^{A_i}(A_i) = (1 + \epsilon)\epsilon n / 2$. We obtain
$$\sum_{j \in A_i} \psi_j > (1 - \epsilon) n \epsilon = \frac{(1 - \epsilon) }{(1 + \epsilon)} 2 C^{A_i}(A_i). $$
Thus, the core property is violated by a $1/2 + \epsilon'$ factor for set $A_i$ and function $C^{A_i}$, and  for any constant $\epsilon' > 0$ by picking $\epsilon$ sufficiently small.
\end{proof}

\section{Approximating the Shapley Value from Samples}
\label{s:shapleyFunction}

We turn our attention to the \scs \ problem in the context of the Shapley value. Since the properties (axioms) of the Shapley value are over elements and not sets, there is no simple relaxation of the Shapley value where the properties hold ``probably" over $\CD$ as we had for the core. However, since the Shapley value exists and is unique for all functions, a natural relaxation is to simply approximate this value from samples. 

We begin by observing that there exists a distribution such that Shapley value can be approximated arbitrarily well from samples. However, in this paper, we are motivated by applications where the algorithm does not control the distribution over the samples, but where the samples are drawn from some ``natural" distribution. Thus, the distributions we consider in this section are the uniform distribution, and more generally product distributions, which are the standard distributions studied in the learning literature for combinatorial functions \cite{balcan2011learning,balcan2012learning,feldman2014learning,
feldman2014optimal}. It is easy to see that we need some restrictions on the distribution $\CD$  (for example, if the empty set if drawn with probability one, the Shapley value cannot be approximated). 

In the case of submodular functions with bounded curvature, we prove a tight approximation bound in terms of the curvature when samples are drawn from bounded product distributions.  However, we show that the Shapley value cannot be approximated from samples even for coverage functions (which are a special case of submodular functions) and the uniform distribution. Since coverage functions are learnable from samples, this implies the counter-intuitive observation that learnability does not imply that the Shapley value is approximable from samples. We begin by formally defining $\alpha$-approximability of the Shapley value in the statistical setting.

\begin{definition}  An algorithm $\alpha$-approximates, $\alpha \in (0,1]$, the Shapley value of  a family of cost functions $\mathcal{C}$ over distribution $\CD$, if, for all $C \in \mathcal{C}$ and all $\delta > 0$, given $\poly(n, \nicefrac{1}{\delta}, \nicefrac{1}{1-\alpha})$ samples from $\CD$, it computes Shapley value estimates $\tilde{\phi}_C$ such that for
\begin{itemize}
\item positive bounded Shapley value, if $\phi_i \geq 1 / \poly(n)$, then   $\alpha \phi_i \leq \tilde{\phi}_i \leq \frac{1}{\alpha}\phi_i$;
\item negative bounded Shapley value, if $\phi_i \leq - 1 / \poly(n)$, then   $\frac{1}{\alpha} \phi_i \leq \tilde{\phi}_i \leq \alpha\phi_i$;
\item small Shapley value, if $ |\phi_i| < 1 / \poly(n) $, then $|\phi_i - \tilde{\phi}_i| = o(1) $ .  
\end{itemize} 
 for all $i \in N$ with probability at least $1 - \delta$ over both the samples and the choices made by the algorithm. 
\end{definition}

\paragraph{Controlling $\CD$.} We begin by noting that there exists a distribution $\CD$ such that the Shapley value can be approximated arbitrarily well for bounded functions. Other sampling methods  have previously been suggested (\cite{bachrach2010approximating,fatima2008linear,liben2012computing,
mann1960values}), but the samples $(S, C(S))$ used in these methods are not i.i.d. and the value query model is assumed. 

\begin{definition}
The \emph{Shapley distribution} $\CD^{sh}$  is the distribution which first picks a size $j \in \{0, \ldots, n\}$ uniformly at random and then draws a uniformly random set of size $j$.
\end{definition}
Let $\CS^j_{i}$ and $\CS^j_{-i}$ be the collections of samples of size $j$ containing element $i$ and not containing it respectively. Define $\avg(\CS) := (\sum_{S \in \CS} C(S)) / |\CS| $ to be the average value of the samples in $\CS$. Consider the following cost allocation:
$$\tilde{\phi}_i = \sum_{j=1}^n \avg(\CS^j_{i}) - \sum_{j=0}^{n-1} \avg(\CS^j_{-i})$$
 When the distribution is the Shapley distribution $\CD^{sh}$, the expected value of this cost allocation is the Shapley value and concentration bounds kick in. 

\begin{restatable}{rPro}{pcontrolledD}
\label{p:controlledD} The Shapley value is $(1 - \epsilon)$-approximable, for any constant $\epsilon > 0$, over the Shapley distribution $\CD^{sh}$. 
\end{restatable}
\begin{proof}
Recall the definition of the Shapley value $\phi$ and observe that
\begin{align*}
\phi_i  & = \EU[\sigma][C(S_{\sigma < i} \cup\{i\}) - C(S_{\sigma < i})]\\
& = \EU[S:|S| \sim \mathcal{U}(\{1,\ldots,n\}), i \in S][C(S)] -   \EU[S:|S| \sim \mathcal{U}(\{0,\ldots,n-1\}), i \not  \in S][C(S)] \\ 
& = \sum_{j=1}^n \E[\avg(\CS^j_{i})] - \sum_{j=0}^{n-1} \E[\avg(\CS^j_{-i})] \\
& = \E[\tilde{\phi}_i]
\end{align*}
Next, observe that by standard concentration bounds and a sufficiently large polynomial number of samples drawn from the Shapley distribution $\CD^{sh}$, $m/\poly(n)$ samples are in $\CS^j_{i}$ and $\CS^j_{-i}$ for all $j$ and $i$. Then, by Hoeffding's inequality,
$$\Pr\left[\left|\avg(\CS^j_{i}) - \E[\avg(\CS^j_{i})]\right| \geq \frac{\epsilon |\phi_i| }{ 2n}\right] \leq 2e^{-\frac{m(\epsilon |\phi_i|)^2}{\poly(n)}}$$
By a union bound, we have
$$\Pr\left[\tilde{\phi}_i - \left(\left|\sum_{j=1}^n \E[\avg(\CS^j_{i})] - \sum_{j=0}^{n-1} \E[\avg(\CS^j_{-i})] \right) \right| \geq \epsilon|\phi_i|\right] \leq 2e^{-\frac{m (\epsilon|\phi_i|)^2}{\poly(n)}}.$$
We get that either $(1 - \epsilon) \phi_i \leq \tilde{\phi}_i \leq  (1+\epsilon) \phi_i$ if $\phi_i > 0$ or $(1 + \epsilon) \phi_i \leq  \tilde{\phi}_i  \leq (1-\epsilon) \phi_i$ if $\phi_i < 0$ with probability $1 - 2e^{-\frac{m \epsilon^2}{\poly(n)}}$, if $|\phi_i| \geq 1/ \poly(n)$.

If $|\phi_i| = o(1/\poly(n))$, we bound the first inequality by $\epsilon/ 2n$ instead of $\epsilon |\phi_i| / 2n$ and obtain
$|\phi_i  - \tilde{\phi}_i | = o(1) $ with  probability $1 - 2e^{-\frac{m\epsilon^2}{\poly(n)}}$.
\end{proof}

A known method to estimate an expected value according to some distribution while given samples from another distribution is called \emph{importance sampling}. Importance sampling reweighs samples according to their probabilities of being sampled. Although the above method achieves accurate estimates with a sufficiently large number of samples,  the number of samples required may be exponential, see  Theorem~\ref{thm:lower}.

\subsection{Submodular functions with bounded curvature}

We consider submodular functions with bounded curvature, a common assumption in the submodular maximization literature \cite{iyer2013submodular,iyer2013curvature,sviridenko2015optimal,vondrak2010submodularity}. We show that the Shapley value of these functions is approximable from samples, for which we derive a tight bound. Intuitively, the curvature of a submodular function bounds by how much the marginal contribution of an element can decrease. This property is useful since the Shapley value of an element can be written as a weighted sum of its marginal contributions over all sets.

\begin{definition} A monotone submodular function $C$ has curvature $\kappa \in [0,1]$ if $C_{N \setminus \{i\}}(i) \geq (1-\kappa) C(i)$ for all $i \in N$. This curvature is bounded if $\kappa < 1$. 
\end{definition}
An immediate consequence of this definition is that $C_S(i) \geq (1-\kappa) C_T(i)$ for all $S,T$ such that $i \not \in S \cup T$ by monotonicity and submodularity.  
The main idea for the approximation is that the expected marginal contribution of an element $i$ to a random set approximates the Shapley value of $i$ by the curvature property. We use the same tool $\tilde{v}_i$  to estimate expected marginal contributions $v_i = \E_{S \sim \CD| i \not \in S}[C_S(i)]$ as for additive functions. Recall that $\tilde{v}_i = \avg(\CS_i) - \avg(\CS_{-i})$ is the difference between the average value of samples containing $i$ and the average value of samples not containing $i$.

 \begin{restatable}{rThm}{tcurv}
\label{t:curv}
Monotone submodular functions with bounded curvature $\kappa$ have Shapley value that is $\sqrt{1-\kappa}-\epsilon$ approximable from samples over the uniform distribution and $1 - \kappa-\epsilon$ approximable over any bounded product distribution  for any constant $\epsilon >0$.
\end{restatable}
  First,  the Shapley value of monotone functions is non-negative since marginal contributions are non-negative by monotonicity. Next, if a Shapley value $\phi_i$ is ``small" ($o(1/\poly(n))$), then $C_S(i)$ is small for all $S$ by the curvature property, implying that $v_i = \E_{S \sim \CD| i \not \in S}[C_S(i)]$ is small as well. By Lemma~\ref{l:concentration}, a generalization of Lemma~\ref{l:addconcentration} showing that $\tilde{v}_i$ is a good estimate of $v_i$, and with $\epsilon = 1/n$, $|v_i - \tilde{v}_i| = o(1)$. With $\tilde{\phi}_i =  \tilde{v}_i$, we then obtain $|\tilde{\phi}_i - \phi_i| = o(1)$.

The interesting case is positive bounded Shapley value $\phi_i$, which we assume for the rest of the analysis. We first show a $1 - \kappa$ approximation for product distributions, which is a straightforward application of the curvature property combined with Lemma~\ref{l:concentration}. Consider the algorithm which computes $\tilde{\phi}_i = \tilde{v}_i$. Note that 
$$\phi_i =\EU[\sigma][C_{A_{\sigma<i}}(i)] \geq (1-\kappa) v_i > \frac{1-\kappa}{1 +\epsilon}   \tilde{v}_i > (1- \kappa-\epsilon)\tilde{v}_i$$
where the first inequality is by curvature and the second by Lemma~\ref{l:concentration}. Similarly, for the other direction, $\phi_i\leq  v_i / (1-\kappa) <  \tilde{v}_i/(1-\kappa -\epsilon)$. The $\sqrt{1-\kappa}$ result is the main technical component of this proof. We begin with a technical overview. 
\begin{enumerate}
\item Denote the uniform distribution over all sets of size $j$ by  $\mathcal{U}_{j}$. Lemma~\ref{l:decreasing} shows that the expected marginal contribution $\E_{S \sim \mathcal{U}_{j}| i \not \in S}[C_S(i)]$ of $i$ to a uniformly random set $S$ of size $j$ is decreasing in $j$, which is by submodularity.
\item Consider $L := (1-\epsilon)n/2$, $H := (1+\epsilon)n/2$. Lemma~\ref{l:vi} shows that $(1+ \epsilon) \cdot \E_{S \sim \mathcal{U}_{L}| i \not \in S}[C_S(i)] \geq v_i$ and $(1- \epsilon)\cdot \E_{S \sim \mathcal{U}_{H}| i \not \in S}[C_S(i)]\leq v_i$, which is because a uniformly random set $S$ is likely to have size between $L$ and $H$ and by submodularity.
\item Combining these two lemmas, roughly \emph{half} of the terms (when $j \leq L$) in the summation $\phi_i = ( \sum_{j=0}^{n-1} \E_{S \sim \mathcal{U}_{j}| i \not \in S}[C_S(i)])/n$ are greater than $v_i$ and the other half (when $j \geq H$) of the terms are smaller. This is the main observation for the improvement from $1 - \kappa$.
\item The above and curvature imply that $(1/2 + (1 - \kappa)/2) v_i \leq   \phi_i \leq (1/2 + 1/(2(1-\kappa))v_i $.
\item By scaling $v_i$ to obtain the best approximation possible with the previous inequality, we obtain a $\sqrt{1-\kappa}$  approximation.
\end{enumerate}

Let $$\tilde{\phi}_i = \frac{2-\kappa}{2\sqrt{1-\kappa}} \cdot \tilde{v}_i$$ be the estimated Shapley value. Denote by $\mathcal{U}_j$ the uniform distribution over all sets of size $j$, so $$\phi_i = \EU[\sigma][C_{A_{\sigma < i}}(i)] = \frac{1}{n}\sum_{j = 0}^{n-1}\EU[S \sim \mathcal{U}_{j}| i \not \in S][C_S(i)] .$$   The main idea to improve the loss from $1-\kappa$ to $\sqrt{1 - \kappa}$ is to observe that $v_i$ can be a factor $1-\kappa$ away from the contribution $\E_{S \sim \mathcal{U}_{j_l}| i \not \in S}[C_S(i)]$ of $j$ to sets of low sizes $j_l \leq L := (1-\epsilon') \cdot n / 2$ or $1-\kappa$ away from its contribution $\E_{S \sim \mathcal{U}_{j_h}| i \not \in S}[C_S(i)]$ to sets of high sizes $j_h \geq H := (1+\epsilon') \cdot n / 2 $, but not both, otherwise the curvature property would be violated as illustrated in Figure~\ref{f:curv}. The following lemma shows that $\E_{S \sim \mathcal{U}_{j}| i \not \in S}[C_S(i)]$ is decreasing in $j$ by submodularity.

\begin{figure}
\centering
\includegraphics[width = 0.6\linewidth]{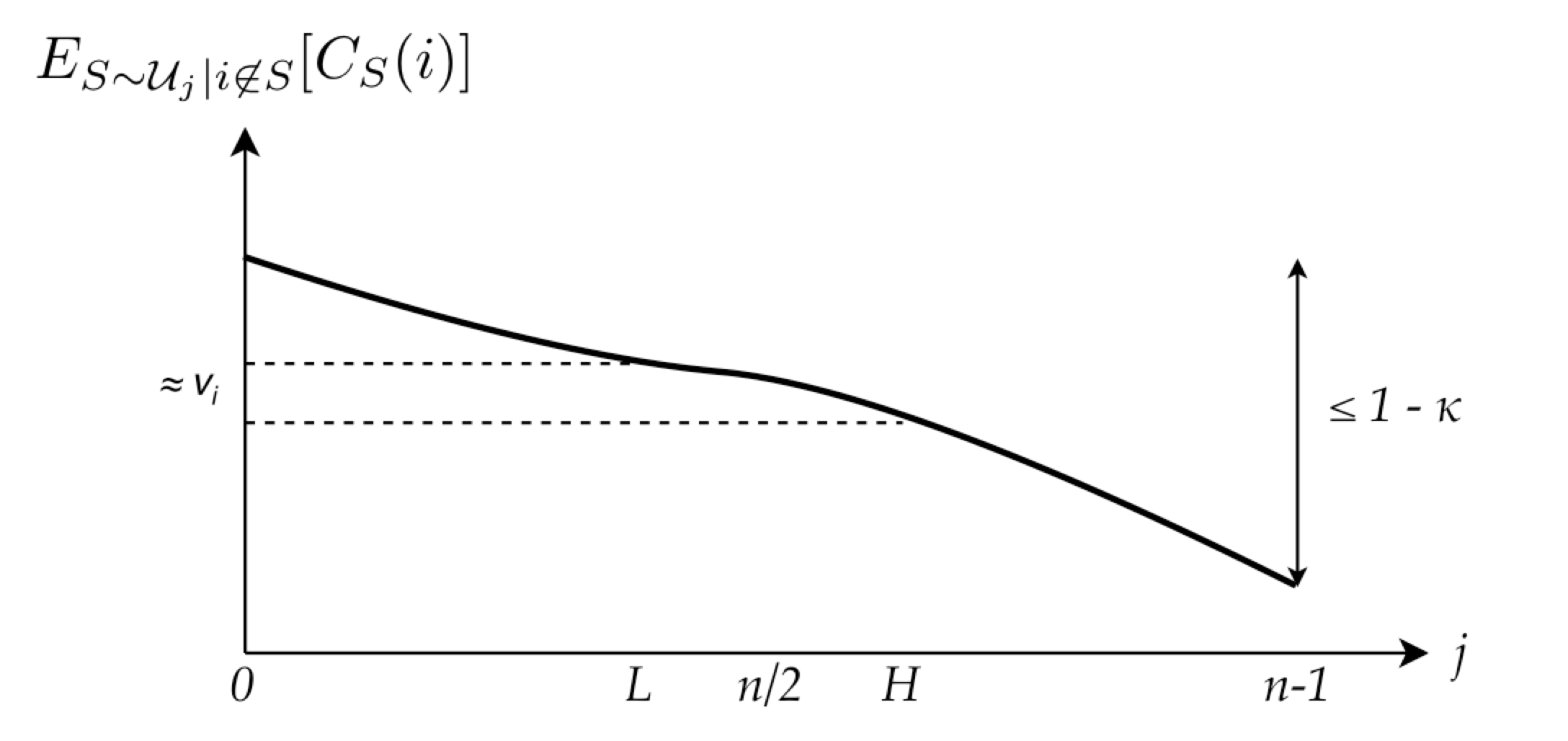}
\caption{The expected marginal contribution $E_{S \sim \mathcal{U}_{j}| i \not \in S}[C_S(i)]$ of an element $i$ to a set of size $j$ as a function of $j$. The curvature property implies that any two points are at most a factor $1-\kappa$ from each other. Lemma~\ref{l:decreasing} shows that it is decreasing. Lemma~\ref{l:vi} shows that the expected marginal contribution $v_i$ of $i$ to a uniformly random set is approximately between $E_{S \sim \mathcal{U}_{L}| i \not \in S}[C_S(i)]$ and $E_{S \sim \mathcal{U}_{H}| i \not \in S}[C_S(i)]$. The Shapley value of $i$ is the average value of this expected marginal contribution over all integers $j \in \{0, \ldots, n-1\}$.}
\label{f:curv}
\end{figure}

\begin{lemma} 
\label{l:decreasing}
Let $C$ be a submodular function, then for all $j \in \{0, \ldots, n-1\}$ and all $i \in N$,
$$\EU[S \sim \mathcal{U}_{j}| i \not \in S][C_S(i)] \geq \EU[S \sim \mathcal{U}_{j+1}| i\not  \in S][C_S(i)] $$
\end{lemma}
\begin{proof}
By submodularity,
$$\sum_{S : |S| = j,  i \not \in S} C_{S}(i) \geq \sum_{S : |S| = j,  i \not \in S} \frac{1}{n -j -1}\sum_{i' \not \in S \cup \{i\}} C_{S \cup \{i'\}}(i). $$
 In addition, observe that by counting in two ways,
$$\sum_{S : |S| = j,  i \not \in S} \sum_{i' \not \in S \cup \{i\}} C_{S \cup \{i'\}}(i)= (j+1) \sum_{S: |S| = j+1, i \not \in S}C_{S}(i).$$
By combining the two previous observations,
\begin{align*}
\EU[S \sim \mathcal{U}_{j}| i \not \in S][C_S(i)] & = \frac{1}{|\{S : |S| = j,  i \not \in S\} |}\sum_{S : |S| = j,  i \not \in S} C_{S}(i)\\
 &\geq \frac{1}{\binom{n-1}{j}} \frac{j+1}{n-j-1}\sum_{S: |S| = j+1, i \not \in S}C_{S}(i) \\
 &=  \frac{1}{|\{S : |S| = j+1,  i \not \in S\} |}\sum_{S: |S| = j+1, i \not \in S}C_{S}(i) \\
&= \EU[S \sim \mathcal{U}_{j+1}| i\not  \in S][C_S(i)] 
\end{align*}
\end{proof}
The next lemma shows that for $j$ slightly lower than $n/2$, the expected marginal contribution $v_i$ of element $i$ to a random set cannot be much larger than $\E_{S \sim \mathcal{U}_{j}| i \not \in S}[C_S(i)]$, and similarly for $j$ slightly larger than $n/2$, it cannot be much smaller.
\begin{lemma}
\label{l:vi}
Let $C$ be a submodular function, then for all $i \in N$, $$\left(1 +  \frac{e^{-\frac{\epsilon n}{6}}}{1 - \kappa} \right) \cdot \EU[S \sim \mathcal{U}_{L}| i \not \in S][C_S(i)] \geq v_i \geq \left(1 - e^{-\frac{\epsilon n}{6}}\right) \cdot \EU[S \sim \mathcal{U}_{H}| i \not \in S][C_S(i)].$$
\end{lemma}
\begin{proof}
By Chernoff bound,  $L \leq |S|$ and $|S| \leq H$ with probability at least $1-e^{-\frac{\epsilon n}{6}}$ each  for $S$ drawn from the uniform distribution. Denote the uniform distribution over all sets by $\mathcal{U}$. So, 
\begin{align*}
v_i & = \sum_{j=0}^{n-1} \Pr_{S \sim \U| i \not \in S}[|S| = j  ] \cdot \EU[S \sim \mathcal{U}_{j}| i \not \in S][C_S(i)] &\\
& \geq \sum_{j=0}^{H} \Pr_{S \sim \U| i \not \in S}[|S| = j] \cdot \EU[S \sim \mathcal{U}_{j}| i \not \in S][C_S(i)] &\\
& \geq \Pr_{S \sim \U| i \not \in S}\left[|S| \leq H \right]\cdot \EU[S \sim \mathcal{U}_{H}| i \not \in S][C_S(i)] & \text{Lemma~\ref{l:decreasing}} \\
& \geq (1-e^{-\frac{\epsilon n}{6}}) \cdot \EU[S \sim \mathcal{U}_{H}| i \not \in S][C_S(i)] . &
\end{align*}
Similarly,
\begin{align*}
v_i & = \sum_{j=0}^{n-1} \Pr_{S \sim \U| i \not \in S}[|S| = j ] \cdot \EU[S \sim \mathcal{U}_{j}| i \not \in S][C_S(i)]  & \\
& \leq \Pr_{S \sim \U| i \not \in S}\left[|S| < L \right] \cdot C(i) + \Pr_{S \sim \U|i \not \in S}\left[|S| \geq L \right] \cdot \EU[S \sim \mathcal{U}_{L}| i \not \in S][C_S(i)] & \text{Lemma~\ref{l:decreasing}} \\
& \leq e^{-\frac{\epsilon n}{6}} \cdot  \frac{1}{1-\kappa} \cdot \EU[S \sim \mathcal{U}_{L}| i \not \in S][C_S(i)] + \EU[S \sim \mathcal{U}_{L}| i \not \in S][C_S(i)]  & \text{curvature} \\
\end{align*}
\end{proof}
We are now ready to prove Theorem~\ref{t:curv}:
\begin{align*}
\phi_i & = \frac{1}{n}\sum_{j=0}^{n-1} \EU[S \sim \mathcal{U}_{j}| i \not \in S][C_S(i)] &\\
& = \frac{1}{n} \left(\sum_{j=0}^{H -1} \EU[S \sim \mathcal{U}_{j}| i \not \in S][C_S(i)] + \sum_{i = H}^{n-1} \EU[S \sim \mathcal{U}_{j}| i \not \in S][C_S(i)]  \right) & \\
& \leq \frac{1+\epsilon'}{2} \cdot C(i) + \frac{1}{2} \cdot \EU[S \sim \mathcal{U}_{H}| i \not \in S][C_S(i)] & \text{Lemma~\ref{l:decreasing}} \\
& \leq \frac{1+\epsilon'}{2(1-\kappa)} \cdot v_i + \frac{1}{2(1-e^{-\frac{\epsilon' n}{6}})} \cdot v_i & \text{curvature and Lemma~\ref{l:vi}} \\
& \leq \left(\frac{2- \kappa}{2(1-\kappa)} + c_1 \epsilon' \right) \cdot v_i  \\
& \leq \left(\frac{2- \kappa}{2(1-\kappa)} + c_2 \epsilon'\right) \cdot \tilde{v}_i & \text{Lemma~\ref{l:concentration}} \\
& = \left(\frac{1}{\sqrt{1-\kappa} - c_3 \epsilon'} \right) \cdot \tilde{\phi}_i & \text{ definition of } \tilde{\phi}_i \\
\end{align*}
for some constants $c_1, c_2, c_3$ and let $\epsilon' = \epsilon / c_3 $ to obtain the desired result for any $\epsilon$. 

Similarly,
\begin{align*}
\phi_i & = \frac{1}{n} \sum_{j=0}^{L} \EU[S \sim \mathcal{U}_{j}| i \not \in S][C_S(i)]  + \frac{1}{n} \sum_{j = L + 1}^{n-1} \EU[S \sim \mathcal{U}_{j}| i \not \in S][C_S(i)] & \\
& \geq \frac{1-\epsilon'}{2}  \cdot \EU[S \sim \mathcal{U}_{L}| i \not \in S][C_S(i)]  + \frac{1}{2} (C(N) - C(N \setminus\{i\}) & \text{Lemma~\ref{l:decreasing}} \\
& \geq \frac{1-\epsilon'}{2\left(1 +  \frac{e^{-\frac{\epsilon' n}{6}}}{1 - \kappa} \right)} \cdot v_i + \frac{1-\kappa}{2} \cdot v_i & \text{Lemma~\ref{l:vi} and curvature} \\
& \geq \left(\frac{2- \kappa}{2} - c_1 \epsilon' \right) v_i \\
& \geq \left(\frac{2- \kappa}{2} - c_2 \epsilon' \right) \tilde{v}_i & \text{Lemma~\ref{l:concentration}} \\
&= \left(\sqrt{1-\kappa} - c_3 \epsilon'\right) \tilde{\phi}_i & \text{ definition of } \tilde{\phi}_i \\
\end{align*} 
\qed

We show that this approximation is optimal. We begin with a general lemma to derive information theoretic  inapproximability results for the Shapley value. This lemma shows that if there exists two functions in $\CC$ that cannot be distinguished from samples with high probability and that have an element with Shapley value which differs by an $\alpha^2$ factor, then $\CC$ does not have a Shapley value that is  $\alpha$-approximable from samples.

\begin{lemma}
\label{l:imp}
Consider a family of cost functions $\mathcal{C}$,  a constant $\alpha \in (0, 1)$, and assume there exist $C^1, C^2 \in \mathcal{C}$ such that: 
\begin{itemize}
\item \textbf{Indistinguishable from samples.} With probability $1 - O(e^{- \beta  n})$ over $S \sim \mathcal{D}$ for some constant $\beta > 0$,  $$C^1(S) = C^2(S).$$ 
\item \textbf{Gap in Shapley value.} There exists $i \in N$ such that $$\phi^{C^1}_i  < \alpha^2 \phi^{C^2}_i.$$ 
\end{itemize}
Then,  $\CC$ does not have  Shapley value that is  $\alpha$-approximable from samples over $\CD$.
\end{lemma}
\begin{proof}
 By a union bound, given $m$ sets $S_1, \ldots, S_m$ drawn i.i.d. from $\mathcal{D}$ with $m$ polynomial in $n$, $C^1(S_j) = C^2(S_j)$ for all $S_j$ with probability $1 - O(e^{- \beta n})$.

Let $C = C^1$ or $C = C^2$ with probability $1/2$ each. Assume the algorithm is given $m$ samples such that  $C^1(S_j) = C^2(S_j) $ and consider its (possibly randomized) choice $\tilde{\phi}_i$. Note that  $\tilde{\phi}_i$ is independent of the randomization of $C$ since $C^1$ and $C^2$ are indistinguishable to the algorithm. Since $\phi^{C^1}_i / \phi^{C^2}_i < \alpha^2$, $\tilde{\phi}_i$ is at least a factor $\alpha$ away from $\phi^{C}_i$ with probability at least $1/2$ over the choices of the algorithm and $C$. Label the cost functions so that $\tilde{\phi}_i$ is at least a factor $\alpha$ away from $\phi^{C^1}_i$ with probability at least $1/2$ over the choices of the algorithm. Thus, with $\delta = 1/4$, there exists no algorithm such that for all $C' \in \{C_1, C^2\}$,   $\alpha \cdot \phi^{C'}_i \leq \tilde{\phi}^{C'}_i \leq \frac{1}{\alpha} \cdot \phi^{C'}_i $ with probability at least $3/4$ over both the samples and the choices of the algorithm. 
\end{proof}

We obtain the inapproximability result by constructing two such functions.

\begin{restatable}{rThm}{tlbcurv} 
\label{t:lbcurv}
For every $\kappa < 1$, there exists a hypothesis class of submodular functions with curvature $\kappa$ that have Shapley value that is not $\sqrt{1-\kappa} + \epsilon$-approximable from samples over the uniform distribution, for every constant $\epsilon > 0$.
\end{restatable}
\begin{proof}

\begin{figure}
\captionsetup{width=0.85\textwidth}
\centering
\begin{subfigure}{.5\textwidth}
  \centering
  \includegraphics[width=1.0\linewidth]{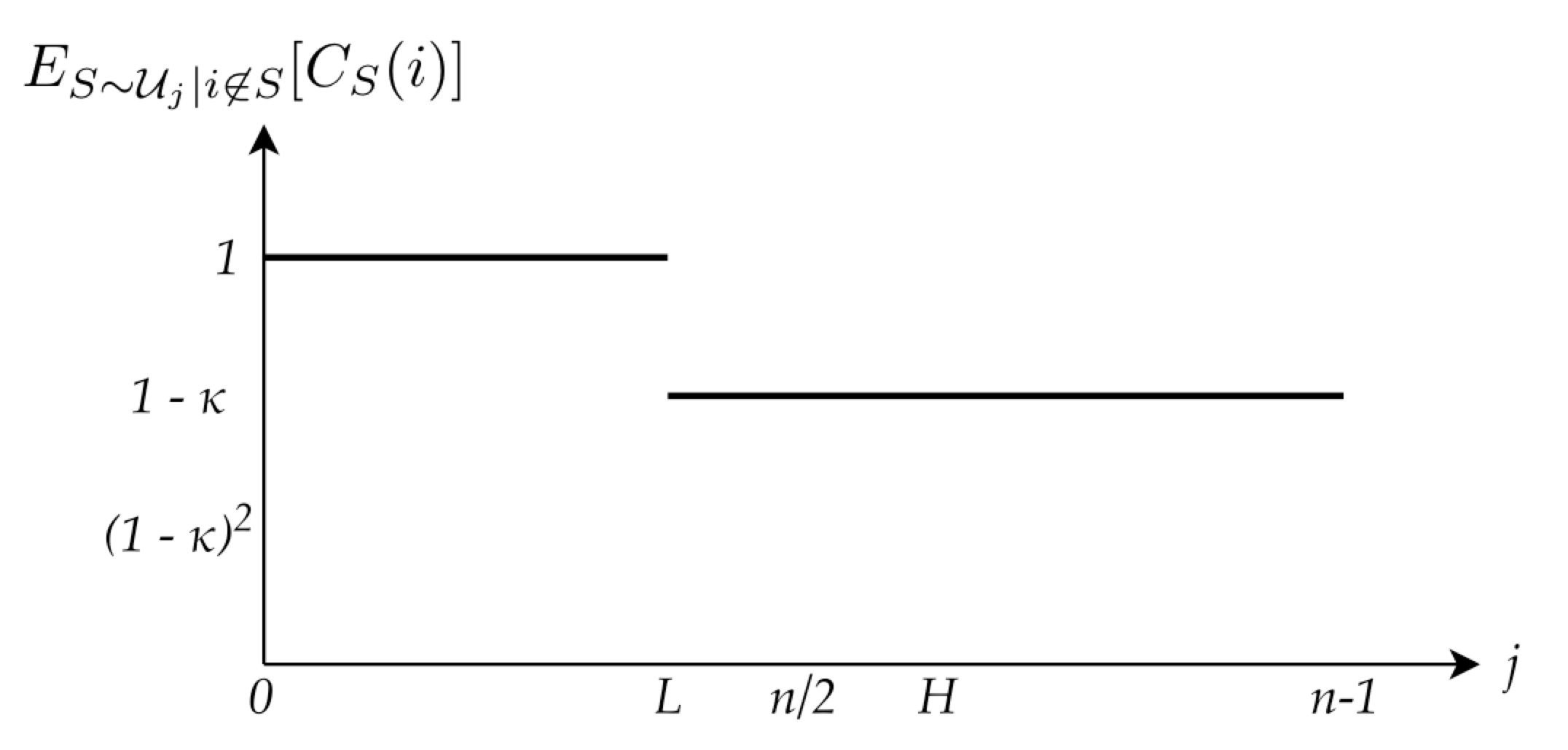}
  \caption{}
   \label{f:1}	
\end{subfigure}%
\begin{subfigure}{.5\textwidth}
  \centering
  \includegraphics[width=1.0\linewidth]{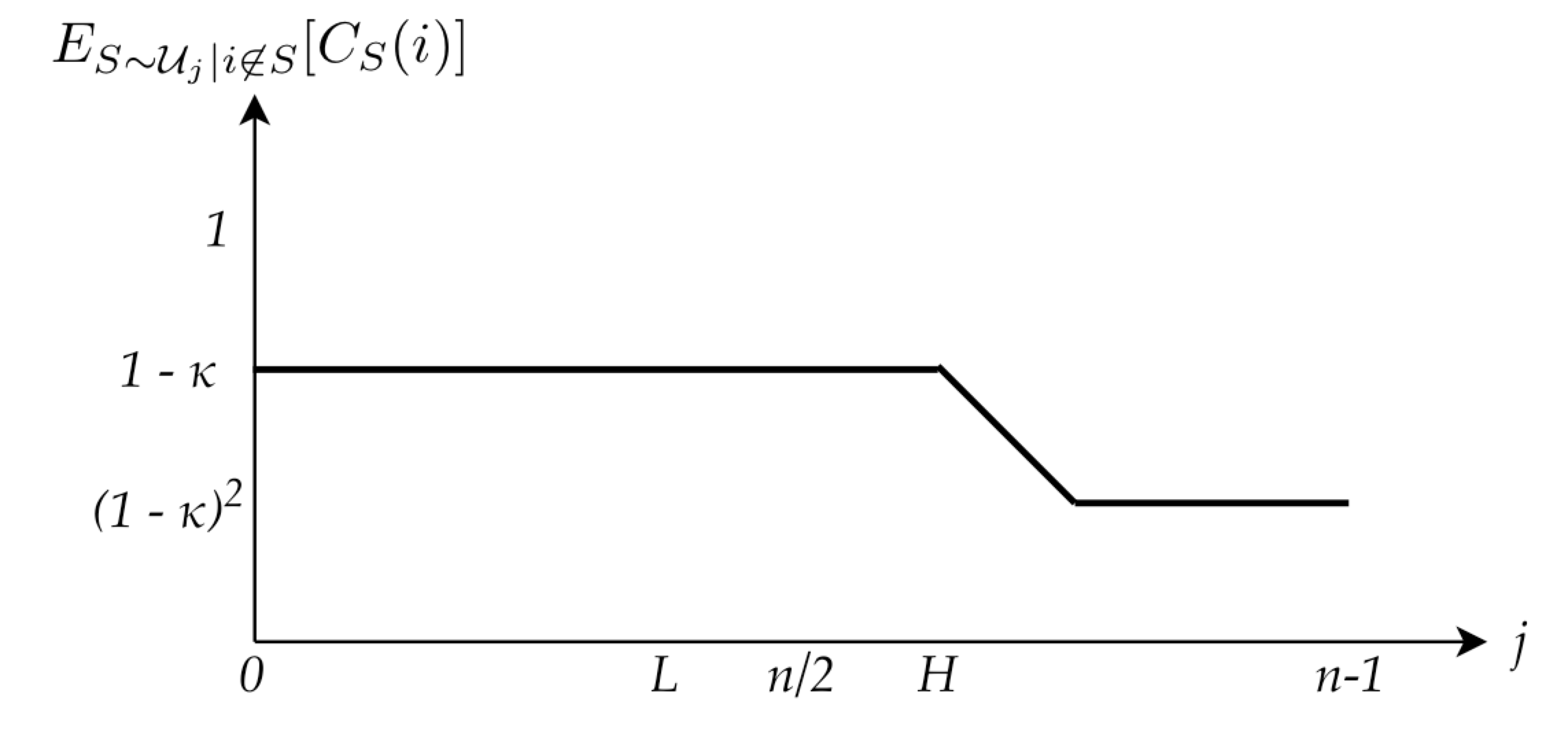}
  \label{f:2}
  \caption{}
\end{subfigure}
\caption{The marginal contributions $C^1_S(i^{\star})$, (a), and $C^2_S(i^{\star})$, (b), of $i^{\star}$ to a set $S$ of size $j$. }
\label{f:lb}
\end{figure}
We first give a technical overview. 
\begin{itemize}
\item We construct two functions $C^1$ and $C^2$ which are indistinguishable from samples but that have an element $i^{\star}$ for which its marginal contribution differs by a factor of $1- \kappa$ for the two functions and then Lemma~\ref{l:imp} concludes the proof. 

\item The expected marginal contribution $\E_{S \sim \mathcal{U}_{j}| i^{\star} \not \in S}[C_S(i^{\star})]$ for both of these functions is illustrated in Figure~\ref{f:lb} as a function of $j$. Informally, $\E_{S \sim \mathcal{U}_{j}| i^{\star} \not \in S}[C_S(i^{\star})]$ is equal for both functions between $L$ and $H$ to obtain indistinguishability from samples since samples are of size between $L$ and $H$ with high probability. Combining this constraint with the submodular and curvature constraints, the gap between $\E_{S \sim \mathcal{U}_{j}| i^{\star} \not \in S}[C^1_S(i^{\star})]$ and $\E_{S \sim \mathcal{U}_{j}| i^{\star} \not \in S}[C^2_S(i^{\star})]$ is maximized for all $j < L$ and $j > H$ to maximize the gap in the Shapley value for $i^{\star}$.
\end{itemize}

  These two functions have a simpler definition via their marginal contributions, so we start by defining them in terms of these marginal contributions and we later give their formal definition to show that they are well-defined. The marginal contributions of $i^{\star}$ are illustrated in Figure~\ref{f:lb}.
$$C^1_S(i) = \begin{cases} 1 & \text{ if } |S| < L \\ 1 - \kappa & \text{ otherwise} \end{cases}$$
$$C^2_S(i^{\star}) = \begin{cases} 1 - \kappa & \text{ if } |S| \leq H \\ 1- \kappa - (|S| - H) \cdot \frac{1 - \kappa - (1 - \kappa)^2}{\sqrt{n}} & \text{ if } H < |S| \leq H+ \sqrt{n} \\ (1 - \kappa)^2 & \text{ otherwise} \end{cases}$$
For $i \neq i^{\star}$:
$$C^2_S(i) = \begin{cases} \frac{L - (1 - \kappa)}{L-1} & \text{ if } |S| < L-1 \text{ or } (|S| = L - 1 \text{ and } i^{\star} \in S) \\ 1 - \kappa & \text{ if } (|S| = L - 1 \text{ and } i^{\star} \not \in S) \text{ or } L \leq |S| \leq H \text{ or } \\ & \ \ \ (H \leq |S| \leq H + \sqrt{n} \text{ and } i^{\star} \not \in S) \\ 1- \kappa - \frac{1 - \kappa - (1 - \kappa)^2}{\sqrt{n}} & \text{ otherwise }    \end{cases}$$
The formal definitions of the functions are
$$C^1(S) = \begin{cases}  |S|  & \text{ if } |S| < L \\ L + (|S| - L) \cdot (1 - \kappa) & \text{ otherwise} \end{cases}$$
and
$$C^2(S) =  \begin{cases} 1_{i^{\star} \in S}  \cdot (1 - \kappa) + (|S| - 1_{i^{\star} \in S}) \cdot  \frac{L - (1 - \kappa)}{L-1}  & \text{if } |S| < L \\ 
                   L + (|S| - L) \cdot (1- \kappa) & \text{if } L \leq |S| \leq H \\ & \text{or } (H < |S| \leq H + \sqrt{n} \text{ and } i^{\star} \not \in S) \\
                          L + (|S| - L) \cdot (1- \kappa) & \\ \hspace{0.35cm} +  1- \kappa - (|S| - H) \cdot \frac{1 - \kappa - (1 - \kappa)^2}{\sqrt{n}} & \text{if } H < |S| \leq H + \sqrt{n} \text{ and } i^{\star}  \in S \\
                           L  + (H + \sqrt{n} - L) \cdot (1- \kappa) + 1_{i^{\star} \in S} \cdot (1 - \kappa)^2 &\\ \hspace{0.35cm} + (|S| - 1_{i^{\star} \in S} -  (H + \sqrt{n}))(1 - \kappa - \frac{1 - \kappa - (1 - \kappa)^2}{\sqrt{n}})  & \text{otherwise}
                          \end{cases}$$
The Shapley value of $i^{\star}$ with respect to $C^1$ and $C^2$ is then: 
$$\phi^{C^1}_{i^{\star}} = 1 \cdot  \frac{1 - \epsilon'}{2} + (1- \kappa) \cdot \frac{1 + \epsilon'}{2} \geq \frac{2 - \kappa}{2} - \epsilon$$
and 
$$\phi^{C^2}_{i^{\star}} \leq (1-\kappa) \cdot \frac{1 + \epsilon'}{2} + (1 - \kappa)^2 \cdot \frac{1 - \epsilon'}{2} \leq \frac{(1- \kappa)(2 - \kappa)}{2} + \epsilon$$
for an appropriate choice of $\epsilon'$. Next, by Chernoff bound and a union bound, $L \leq |S| \leq H$ for polynomially many samples $S$ from the uniform distirbution, with probability $1 - e^{-\Omega(n)}$. Thus, $C^1(S) = C^2(S)$ for all samples $S$ with probability $1 - e^{-\Omega(n)}$.

It remains to show that $C^1$ and $C^2$ are submodular with curvature $\kappa$, i.e., for any $S \subseteq T$ and $i \not \in T$, $$C_S(i) \geq C_T(i) \geq (1 - \kappa) C_S(i),$$ which is immediate for $C^1$. Regarding $C^2$, it is also immediate that $C^2_S(i^{\star}) \geq C^2_T(i^{\star}) \geq (1 - \kappa) C^2_S(i^{\star})$. For $i \neq i^{\star}$, observe that $$\frac{L - (1 - \kappa)}{L-1} \leq 1 + \epsilon\hspace{0.5cm} \text{ and } \hspace{0.5cm} 1- \kappa - \frac{1- \kappa - (1 - \kappa)^2}{\sqrt{n}} \geq 1 - \kappa - \epsilon,$$ so $C^2_S(i) \geq C^2_T(i) \geq (1 - \kappa - \epsilon) C^2_S(i)$.
\end{proof}

\subsection{Learnability does not imply approximability of the Shapley value}
Although the Shapley value is approximable for the class of submodular functions with bounded curvature, we show that the Shapley value of coverage (and submodular) functions are not approximable from samples in general.  The impossibility result is information theoretic and is not due to computational limitations. Coverage functions are an interesting class of functions because they are learnable from samples over any distribution \cite{badanidiyuru2012sketching}, according to the PMAC learning model  \cite{balcan2011learning}, which is a generalization of PAC learnability  for real valued functions. In addition, by Theorem~\ref{p:controlledD}, coverage functions have Shapley value that can  efficiently be approximated arbitrarily well in the value query model. Thus, this impossibility result implies that learnability and approximability in the value query model are not strong enough conditions for approximability of the Shapley value from samples.

\begin{theorem}
\label{thm:lower}
 There exists no constant $\alpha > 0$ such that coverage functions have Shapley value that is $\alpha$-approximable from samples over the uniform distribution.
\end{theorem}
\begin{proof}
Partition $N$ into two parts $A$ and $B$ of equal size. Consider the following two functions:
$$ C^1(S) = \begin{cases} 0  & \text{if } S = \emptyset \\ 1 & \text{if } |S \cap A| > 0 \text{,   } |S \cap B| = 0  \\ \frac{1}{\alpha^2} & \text{if } |S \cap A| = 0 \text{,   } |S \cap B| > 0 \\ 1 + \frac{1}{\alpha^2} & \text{otherwise}\end{cases}  \hspace{0.3cm}    C^2(S) = \begin{cases} 0  & \text{if } S = \emptyset \\ 1 & \text{if } |S \cap B| > 0 \text{,   } |S \cap A| = 0  \\ \frac{1}{\alpha^2} & \text{if } |S \cap B| = 0 \text{,   } |S \cap A| > 0 \\ 1 + \frac{1}{ \alpha^2} & \text{otherwise}\end{cases}$$
These functions are coverage functions with $U = \{a, b_1, \ldots, b_{1/ \alpha^2}\}$ and $T_i = \{a\}$ or $T_i = \{b_1, \ldots, b_{1/ \alpha^2}\}$.  By the Chernoff bound (Lemma~\ref{l:chernoff}) with $\delta = 1/2$ and $\mu = n/2$, if $S$ is a sample from the uniform distribution, then $$ \Pr\left[|S \cap A| = 0\right] = \Pr\left[|S \cap B| = 0\right] < \Pr\left[|S \cap B| \leq n / 4\right] \leq e^{-n/16},$$ so $C^1(S) \neq C^2(S)$  with probability at most $2e^{-n/16}$.  It is easy to see that for any $i$, its Shapley value is either $2/n$ or $2/(\alpha^2 n)$ depending on which partition it is in. Combining this with Lemma~\ref{l:imp} concludes the proof.
\end{proof}

\section{Data Dependent Shapley Value}

The general impossibility result for computing the Shapley value from samples arises from the fact that the concept was geared towards the query model, where the algorithm can ask for the cost of any set $S \subseteq N$.  In this section, we develop an analogue that is data- or distribution-dependent. We denote it by $\phi^{C, \mathcal{D}}$  with respect to both $C$ and $\mathcal{D}$.  We define four natural distribution-dependent axioms resembling the Shapley value axioms, and then prove that our proposed value is the unique solution satisfying them. This value can be approximated arbitrarily well in the statistical model for all functions. We start by stating the four axioms.

\begin{definition}
\label{d:ddshapley}
The \emph{data-dependent axioms} for cost sharing functions are: 
\begin{itemize}
\item \textbf{Balance:} $\sum_{i \in N} \phi^{\mathcal{D}}_i = \E_{S \sim \mathcal{D}}[C(S)]$,
\item \textbf{Symmetry:} for all $i$ and  $j$, if $\Pr_{S \sim \mathcal{D}}\left[|S \cap \{i, j\}| =1\right] = 0 $ then $\phi^{\mathcal{D}}_i = \phi^{\mathcal{D}}_j$,
\item \textbf{Zero element:} for all $i$, if
$\Pr_{S \sim \mathcal{D}}\left[i \in S\right] = 0 $ then $\phi^{\CD}_i = 0$,
\item \textbf{Additivity:} for all $i$, if $\mathcal{D}_1$, $\mathcal{D}_2$, $\alpha$, and $\beta$ such that $\alpha + \beta = 1$,
$\phi^{\alpha \mathcal{D}_1 + \beta \mathcal{D}_2}_i = \alpha \phi^{\mathcal{D}_1}_i + \beta \phi^{\mathcal{D}_2}_i$
where $\Pr\left[S \sim \alpha \mathcal{D}_1 + \beta \mathcal{D}_2\right] = \alpha \cdot \Pr\left[S \sim \mathcal{D}_1\right] + \beta \cdot \Pr\left[S \sim \mathcal{D}_2\right]$.
\end{itemize}
\end{definition}
The similarity to the original Shapley value axioms is readily apparent. The main distinction is that we expect these to hold with regard to $\CD$, which captures the frequency with which different coalitions $S$ occur.  Note that we no longer require that $\CD$ has full support over $2^N$.  Interpreting the axioms one by one, the balance property ensures that the expected cost is always accounted for. The symmetry axiom states that if two elements always occur together, they should have the same share, since they are indistinguishable. If an element is never observed, then it should have zero share. Finally costs should combine in a linear manner according to the distribution.  

These axioms are specifically designed to provide some guarantees on the shares of elements to functions with complex interactions where recovery is hard from samples.

%
 We define the \emph{data-dependent Shapley value}:
$$\phi^{\CD}_i := \sum_{S \given i \in S} \Pr\left[S \sim \mathcal{D}\right] \cdot \frac{C(S)}{|S|}.$$
Informally, for all set $S$, the cost $C(S)$ is divided equally between all elements in $S$ and is weighted with the probability that $S$ occurs according to $\CD$. The main appeal of this cost allocation is the following theorem.
\begin{theorem}
\label{t:ddshapley}
The data-dependent Shapley value is the unique value satisfying the four data-dependent axioms.
\end{theorem}
We first show that if there exists a value satisfying the axioms, it must be the data-dependent Shapley value. Then, we show that the data-dependent Shapley value satisfies the axioms, which concludes the proof.

\begin{lemma}
If there exists a value satisfying the four data-dependent Shapley axioms, then this value is the data-dependent Shapley value.
\end{lemma}
\begin{proof}
Define $\CD_S$ to be the distribution such that $\Pr\left[S \sim \CD_S\right] = 1$. Observe that the unique value satisfying the balance, symmetry, and zero element axioms must satisfy 
$$\phi^{\CD_S}_i = \begin{cases} \frac{C(S)}{|S|} &  \text{ if } i\in S \\ 0 & \text{ otherwise.} \end{cases}$$
Since $\CD = \sum_{S} \Pr\left[S \sim \CD\right] \cdot \CD_S$, the unique value satisfying the four axioms must satisfy
$$\phi^{\CD}_i = \sum_{S} \Pr\left[S \sim \CD\right] \cdot \phi^{\CD_S}_i = \sum_{S : i \in S} \Pr\left[S \sim \CD\right] \cdot \frac{C(S)}{|S|}$$
where the first equality is by additivity and the second equality by the above observation.
\end{proof}
\begin{lemma}
The data-dependent Shapley value satisfies the four data-dependent Shapley axioms.
\end{lemma}
\begin{proof} We show that each axiom is satisfied.
\begin{itemize}
\item \textbf{Balance:} By definition,
$\sum_{i \in N} \phi^{\mathcal{D}}_i = \sum_{i \in N}\sum_{S:i \in S} \Pr\left[S \sim \mathcal{D}\right] C(S) /|S|$, then by switching the order of the summations, $$\sum_{i \in N}\sum_{S : i \in S} \Pr\left[S \sim \mathcal{D}\right] \frac{C(S)}{|S|} =   \sum_{S \subseteq N} \sum_{i \in S} \Pr\left[S \sim \mathcal{D}\right] \frac{C(S)}{|S|} = \sum_{S \subseteq N} \Pr\left[S \sim \mathcal{D}\right]  C(S) = \EU[S \sim \CD][C(S)].$$
\item \textbf{Symmetry:} Let $\CS_i = \{S: i \in S, \Pr\left[S \sim \CD\right] > 0\}$. If $\Pr_{S \sim \mathcal{D}}\left[|S \cap \{i, j\}| =1\right] = 0$, then $\CS_{i} = \CS_{j}$ and 
$$\phi^{\mathcal{D}}_i = \sum_{S \in \CS_{i}} \Pr\left[S \sim \mathcal{D}\right]  \frac{C(S)}{|S|} = \sum_{S \in \CS_{j}} \Pr\left[S \sim \mathcal{D}\right] \frac{C(S)}{|S|} = \phi^{\mathcal{D}}_j.$$
\item \textbf{Zero element:} If $\Pr_{S \sim \mathcal{D}}\left[i \in S\right] = 0$, then $\Pr\left[S \sim \CD\right] = 0$ if $i \in S$. Thus,  $\phi^{\CD}_i = 0$.
\item \textbf{Additivity:} By definition of $\phi$ and $\alpha \mathcal{D}_1 + \beta \mathcal{D}_2$, 
\begin{align*}
\phi^{\alpha \mathcal{D}_1 + \beta \mathcal{D}_2}_i & =  \sum_{S \given i \in S} \Pr\left[S \sim \alpha \mathcal{D}_1 + \beta \mathcal{D}\right]  \frac{C(S)}{|S|} 
 \\
& = \alpha \sum_{S: i \in S}  \Pr\left[S \sim \mathcal{D}_1\right]   \frac{C(S)}{|S|} + \beta \sum_{S:i \in S} \Pr\left[S \sim \mathcal{D}_2\right]  \frac{C(S)}{|S|}  \\
& = \alpha \phi^{\mathcal{D}_1}_i + \beta \phi^{\mathcal{D}_2}_i.
\end{align*}
\end{itemize}
\end{proof}

The data-dependent Shapley value can be approximated from samples  with the following empirical data-dependent Shapley value:
$$\tilde{\phi}^{\CD}_i = \frac{1}{m} \sum_{S_j \given i \in S_j} \frac{C(S_j)}{|S_j|}.$$

These estimates are arbitrarily good with arbitrarily high probability. 

\begin{theorem} 
\label{t:apxddshapley}
The empirical data-dependent Shapley value approximates the data-dependent Shapley value arbitrarily well, i.e., $|\tilde{\phi}^{\CD}_i - \phi^{\CD}_i| < \epsilon$ with $\poly(n, 1/\epsilon, 1/\delta)$ samples and with probability at least $1 - \delta$ for any $\delta, \epsilon > 0$.
\end{theorem}
\begin{proof}
Define $X_j = \begin{cases} \frac{C(S_j)}{ |S_j|} & \text{ if } i \in S_j \\ 0 & \text{ otherwise} \end{cases}$
and observe that $(\sum_{j=1}^m X_j)/m = \tilde{\phi}^{\CD}_i$ and\\ $\E[(\sum_{j=1}^m X_j)/m] = \phi^{\CD}_i$. Clearly, $X_j \in [0, b]$ where $b := \max_S C(S) / |S|$, so by Hoeffding's inequality,
$\Pr\left[|\tilde{\phi}^{\CD}_i - \phi^{\CD}_i| \geq \epsilon\right] \leq 2e^{-\frac{2m \epsilon^{2}}{\poly(n)}}$ with $0 < \epsilon < 1$.
\end{proof}

\section{Discussion and Future Work}

We follow a recent line of work that studies classical algorithmic problems from a statistical perspective, where the input is restricted to a collection of samples.  Our results fall into two categories, we give results for approximating the Shapley value and the core and propose new cost sharing concepts that are tailored for the statistical  framework. We use techniques from multiple fields that encompass statistical machine learning, combinatorial optimization, and, of course, cost sharing. The cost sharing literature being very rich, the number of directions for future work are considerable. Obvious avenues include studying other cost sharing methods in this statistical framework, considering other classes of functions to approximate known methods, and improving the sample complexity of previous algorithms. More conceptually, an exciting modeling question arises when designing ``desirable" axioms from data. Traditionally these axioms only depended on the cost function, whereas in this model they can depend on both the cost function and the distribution, providing an interesting interplay.

\newpage
\bibliographystyle{plainnat}
\bibliography{biblio} 

\newpage

\appendix
\section*{Appendix}
\label{s:appendix}

\section{Concentration Bounds}
\label{s:concbounds}

\begin{lemma}[Chernoff Bound]
\label{l:chernoff}
Let $X_1, \dots, X_n$ be independent indicator random variables with values in $\{0,1\}$. Let $X = \sum_{i=1}^n X_i$ and $\mu = \E[X]$. For $0 < \delta < 1$,
$$\Pr\left[X \leq (1-\delta) \mu\right] \leq  e^{- \mu \delta^2 / 2} \ \ \ \text{ and } \ \ \ \Pr\left[X \geq (1+\delta) \mu\right] \leq  e^{- \mu \delta^2 / 3}.$$
\end{lemma}

\begin{lemma}[Hoeffding's inequality]
\label{l:hoeffding}
Let $X_1, \dots, X_n$ be independent random variables with values in $[0,b]$. Let $X = \frac{1}{m}\sum_{i=1}^m X_i$ and $\mu = \E[X]$. Then for every $0 < \epsilon < 1$,
$$\Pr\left[|X - \E[X]| \geq \epsilon\right] \leq 2 e^{- 2 m \epsilon^2 /  b^2}.$$
\end{lemma}

\section{Estimating the Expected Marginal Contribution of an Element}
\label{s:appprelim}
Recall that $\CS_i$ and $\CS_{-i}$ are the collections of samples containing element $i$ and not containing it respectively and that $\avg(\CS) = (\sum_{S \in \CS} C(S)) / |\CS| $ is the average value of the samples in $\CS$. Let $v_i = C(i)$ and $\tilde{v}_i = \avg(\CS_i) - \avg(\CS_{-i})$. 

\lconcentration*
This lemma is a special case of the following stronger lemma.
 \begin{lemma}
\label{l:concentration} The expected marginal contribution  of an element $i$ to a random set from a bounded product distribution $\CD$ not containing $i$ is estimated arbitrarily well by $\tilde{v}_i$, i.e.,  for all $i \in N $ and given $\poly(n, \nicefrac{1}{\delta}, \nicefrac{1}{\epsilon})$ samples,   
\begin{align*}
 & (1 - \epsilon) v_i \leq \tilde{v}_i \leq (1+\epsilon) v_i & \text{ if } & v_i \geq 1 / \poly(n) \\
& |v_i - \tilde{v}_i| \leq  \epsilon            & \text{if } &|v_i| < 1 / \poly(n) \\
 & (1 + \epsilon) v_i \leq \tilde{v}_i \leq  (1-\epsilon) v_i &  \text{ if } &v_i \leq -  1 / \poly(n)
\end{align*} 
   with probability at least $1 - \delta$ for any $\delta> 0$.
\end{lemma}
\begin{proof}
Note that 
$$v_i = \EU[S \sim \CD|i \not \in S][C_S(i)] = \EU[S \sim \CD|i \not \in S][C(S \cup i)] - \EU[S \sim \CD|i \not \in S][C(S)] = \EU[S \sim \CD|i \in S][C(S)] - \EU[S \sim \CD|i \not \in S][C(S)].$$
where the second equality is since $\CD$ is a product distribution. In addition, $\E[\avg(\CS_i)] = \E_{S \sim \CD|i \in S}[C(S)] $ and  $\E[\avg(\CS_{-i})] = \E_{S \sim \CD|i \not \in S}[C(S)] $. Since marginal probabilities of the product distributions are assumed to be bounded from below and above by $1/\poly(n)$ and $1 - 1/\poly(n)$ respectively, $|\CS_{i}| = m / \poly(n)$ and $|\CS_{-i}| = m / \poly(n)$ for all $i$ by Chernoff bound. In addition, $\max_S C(S)$ is  assumed to be bounded by $\poly(n)$. So by Hoeffding's inequality, 
$$\Pr\left(\left|\avg(\CS_i) - \EU[S \sim \CD|i \in S][C(S)]\right| \geq |v_i| \epsilon/2\right) \leq 2 e^{-\frac{  m(|v_i|\epsilon)^2}{\poly(n)}},$$
for $0 < \epsilon < 2/v_i$ and 
$$\Pr\left(\left|\avg(\CS_{-i}) - \EU[S \sim \CD|i \not \in S][C(S)]\right| \geq |v_i| \epsilon / 2\right) \leq 2 e^{-\frac{  m(|v_i|\epsilon)^2}{\poly(n)}}.$$
Thus,
$$\Pr(|\tilde{v}_i - v_i| \geq |v_i| \epsilon)  \leq 2 e^{-\frac{  m(|v_i|\epsilon)^2}{\poly(n)}}$$
and,  either $(1 - \epsilon) v_i \leq \tilde{v}_i \leq  (1+\epsilon) v_i$ if $v_i > 0$ or $(1 + \epsilon) v_i \leq  \tilde{v}_i \leq (1-\epsilon) v_i$ if $v_i < 0$, with probability at least $1 -2 e^{-\frac{  m(|v_i|\epsilon)^2}{\poly(n)}}$. If $|v_i| < 1 / \poly(n) $, we obtain  $|v_i - \tilde{v}_i| < \epsilon$ with a similar analysis without any assumption on $v_i$. Otherwise, the bounds on the estimation hold with probability at least $1 - 2 e^{\frac{m\epsilon^2}{\poly(n)}}$.
\end{proof}

\begin{corollary}
Let $C$ be an additive function such that $C(i) \geq 1 / \poly(n)$. Then, $C$ has Shapley value and a core that are $(1-\epsilon)$-approximable from samples over bounded product distributions for any constant $\epsilon > 0$.
\end{corollary}

\begin{proof} For the Shapley value, it follows  immediately from Lemma~\ref{l:concentration}. Regarding the core, let
 $\psi_i =\tilde{v}_i \cdot \frac{C(N)}{\sum_{j \in N} \tilde{v}_j }$, so roughly $\tilde{v}_i$ but slightly scaled to obtain the balance property, which holds since $\sum_{i \in N} \psi_i = \sum_{i \in N} \tilde{v}_i \cdot \frac{C(N)}{\sum_{j \in N} \tilde{v}_j } = C(N)$. For the approximate core property, first note that $\frac{C(N)}{\sum_{j \in N} \tilde{v}_j } \leq \frac{C(N)}{(1-\epsilon')\sum_{j \in N} C(j) } = \frac{1}{1-\epsilon'}$, so, 
$$(1 - \epsilon) \sum_{i \in S} \psi_i \leq  (1 - \epsilon) \sum_{i \in S}  \tilde{v}_i \cdot \frac{C(N)}{\sum_{j \in N} \tilde{v}_j } \leq  \frac{(1 - \epsilon)(1+\epsilon')}{1 - \epsilon'}\sum_{i \in S}  C(i) \leq  C(S)$$ for $\epsilon'$ picked accordingly small compared to $\epsilon$. 
\end{proof}

\section{VC-Dimension and Rademacher Complexity Review}
\label{s:appcore}

We formally define the VC-dimension and the Rademacher complexity using definitions from \cite{ML}. We begin with the VC-dimension, which is for classes of binary functions. We first define the concepts of restriction to a set and of shattering, which are useful to define the VC-dimension. 

\begin{definition}
(Restriction of $\CH$ to $A$). Let $\CH$ be a class of functions from $\mathcal{X}$ to $\{0,1\}$ and let $A = \{a_1, \ldots, a_m\} \subset \CX$. The restriction of $\CH$ to $A$ is the set of functions from $A$ to $\{0,1\}$ that can be derived from $\CH$. That is,
$$\CH_A =\{(h(a_1), \ldots, h(a_m)) \given h \in \CH\},$$
where we represent each function from $A$ to $\{0,1\}$ as a vector in $\{0,1\}^{|A|}$.
\end{definition}
\begin{definition}(Shattering). A hypothesis class $\CH$ shatters a finite set $A \subset \CX$ if the restriction of $\CH$ to $A$ is the set of all functions from $A$ to $\{0,1\}$. That is, $|\CH_A| = 2^{|A|}$.
\end{definition}
\begin{definition}(VC-dimension). The VC-dimension of a hypothesis class $\CH$ is the maximal size of a set $S \subset \CX$ that can be shattered by $\CH$. If $\CH$ can shatter sets of arbitrarily large size we say that $\CH$ has infinite VC-dimension.
\end{definition}

Next, we define the Rademacher complexity, which is for more complex classes of functions than binary functions, such as real-valued functions. 

\begin{definition}(Rademacher complexity). Let $\mathbf{\sigma}$ be distributed i.i.d. with $\Pr[\sigma_i = 1] = \Pr[\sigma_i = -1] = 1/2$. The Rademacher complexity $R(A)$ of a set of vectors $A \subset \R^m$ is  
$R(A) := \frac{1}{m} \EU[\sigma]\left[\sup_{a \in A} \sum_{i=1}^m \sigma_i a_i\right].$
\end{definition}

\end{document}